\documentclass[11pt,a4paper]{article}
\usepackage{arxiv}

\usepackage[utf8]{inputenc}
\usepackage[T1]{fontenc}

\usepackage[x11names]{xcolor}
\usepackage{amsmath,amsfonts}
\usepackage[caption=false,font=normalsize,labelfont=sf,textfont=sf]{subfig}
\usepackage{url}
\usepackage{graphicx}
\usepackage{cite}
\usepackage{flushend} 
\usepackage{pifont}
\usepackage[usenames,dvipsnames]{xcolor}
\usepackage{hyperref}       
\hypersetup{
    colorlinks=true,
    linkcolor=blue!70!black!70,
    filecolor=mnodea,      
    urlcolor=cyan,
    citecolor=purple!70!black!70,
}

\usepackage{url}            
\usepackage{booktabs}       
\usepackage{amsfonts,amsmath,amssymb}       
\usepackage{amsthm}     
\usepackage{nicefrac}       
\usepackage{microtype}      
\usepackage{graphicx}
\usepackage{enumitem}
\usepackage{tikz} 
\usepackage{tkz-berge}
\usepackage{caption}
\usepackage{pgfplots}
\usepackage{multicol}

\usepackage[ruled,vlined,linesnumbered]{algorithm2e}
\SetKwInput{KwIn}{\small Input}
\SetKwInput{KwInit}{\small Initialization}
\SetKwInput{KwOut}{\small Output}
\SetAlgoSkip{}
\SetAlgoInsideSkip{smallskip}
\SetInd{0.5em}{1.2em}
\SetNlSty{textbf}{}{}

\usepackage{textgreek}
\usepackage{soul}
\usepackage[mode=buildnew]{standalone}
\usepackage{booktabs}
\usepackage{multirow}

\usetikzlibrary{math,spy,arrows,arrows.meta, backgrounds,chains, positioning, shapes.geometric, shapes.symbols, shapes.arrows, patterns, automata, petri, topaths,plotmarks,decorations.fractals,pgfplots.groupplots}%
\pgfplotsset{compat=newest}
\usepgfplotslibrary{fillbetween,colorbrewer}
\definecolor{skyblue1}{rgb}{0.447,0.624,0.812}

\newcommand{\set}[1]{\mathcal{#1}} 
\newcommand{\ie}{\textit{i.e.,~}} 
\newcommand{\eg}{\textit{e.g.,~}} 
\newcommand{\etal}{\textit{et.~al.~}} 
\newcommand{\inneighbor}[1]{\mathcal{N}_{#1}^{\scriptscriptstyle -}}
\newcommand{\outneighbor}[1]{\mathcal{N}_{#1}^{\scriptscriptstyle +}}

\newcommand{\outdegree}[1]{d_{#1}^{\scriptscriptstyle +}}

\newcommand{\harqrc}{{\small HARQ-RC}} 
\newcommand{\arqrc}{{\small ARQ-RC}}  
\newcommand{\arq}{{\small ARQ}}  
\newcommand{\harq}{{\small HARQ}} 
\newcommand{\hharq}{{\small (H)ARQ}} 
\newcommand{\harqratioconsensus}{{\small HARQ}~{\small R}atio~{\small C}onsensus} 
\newcommand{\rcrs}{{\small RC-RS}} 

\DeclareSymbolFont{matha}{OML}{txmi}{m}{it}
\DeclareMathSymbol{\varv}{\mathord}{matha}{118}

\newcommand{\nnreal}{\mathbb{R}_{\geq0}}
\newcommand{\nnint}{\mathbb{Z}_{\geq0}}
\newcommand{\posreal}{\mathbb{R}_{+}}

\newtheorem{thm}{Theorem}

\newtheorem{prop}{Proposition}
\newtheorem{cor}{Corollary}

\theoremstyle{definition}

\newtheorem{rem}{Remark}
\newtheorem{assumption}{Assumption}

\newcommand\Label[1]{&\text{\refstepcounter{equation}(\theequation)\ltx@label{#1}}&}

\definecolor{Blue}{rgb}{0.0, 0.1, 0.6}

\makeatletter
\newcommand{\smallbullet}{} 
\DeclareRobustCommand\smallbullet{%
  \mathord{\mathpalette\smallbullet@{0.5}}%
}
\newcommand{\smallbullet@}[2]{%
  \vcenter{\hbox{\scalebox{#2}{$\m@th#1\bullet$}}}%
}
\makeatother

\begin{document}

\title{\LARGE Exact Average Consensus over Unreliable Directed Networks with HARQ Feedback: Almost-Sure Convergence and Rate Characterization}

\author{Evagoras Makridis, Christoforos N. Hadjicostis, and Themistoklis Charalambous
\thanks{E. Makridis, C. N. Hadjicostis, and T. Charalambous are with the Department of Electrical and Computer Engineering, University of Cyprus, Nicosia, Cyprus. T. Charalambous is also a Visiting Professor with the Department of Electrical Engineering and Automation, Aalto University, Espoo, Finland, and the FinEst Centre for Smart Cities. E-mails: \{{\tt\footnotesize surname.name}\}{\tt\footnotesize @ucy.ac.cy}.}
\thanks{Preliminary results of this work have been published in \cite{makridis2023utilizing,makridis2023harnessing}. The present article substantially extends these conference papers in both modeling and theoretical analysis. In particular, compared with the preliminary conference versions, this paper introduces the retransmission-dependent HARQ error model, develops the augmented stochastic mass-flow representation and establishes its exact sample-path equivalence with the distributed iterations, proves almost-sure convergence to the exact average under finite retransmission limits and permanent packet drops, and derives nonasymptotic convergence-rate, high-probability, and expected consensus-time guarantees.}}

\maketitle

\begin{abstract}
We study discrete-time average consensus over strongly connected directed networks with unreliable communication links, where packet errors trigger Hybrid Automatic Repeat reQuest (HARQ) retransmissions and packets are permanently dropped after a finite retransmission limit. We propose a HARQ-based ratio-consensus algorithm that exploits ACK/NACK feedback to explicitly account for successfully delivered, delayed, and eventually dropped information mass, thereby preserving the exact network-wide average despite asymmetric communication and packet losses. To model the dependence induced by HARQ combining, we describe retransmission errors through conditional decoding-failure probabilities where we simply require the conditional probability of successful transmission to be uniformly bounded away from zero. We construct an augmented sample-path representation of the resulting delayed and dropped information flows and establish its exact equivalence with the proposed distributed iterations at the actual network nodes. Using this representation, we prove that all node ratios converge almost surely to the exact arithmetic average of the initial values. Beyond asymptotic convergence, we derive an explicit geometric bound on the expected consensus error, a high-probability error bound, and a finite upper bound on the expected time required to reach a prescribed consensus accuracy. Numerical experiments illustrate the effect of the channel reliability, HARQ combining gain, and retransmission limit on the convergence--communication trade-off, and compare the proposed method with existing robust ratio-consensus schemes.
\end{abstract}


\section{Introduction}\label{sec:introduction}

Distributed averaging is a fundamental building block in distributed estimation, control, and optimization, where agents collectively compute quantities from locally available information \cite{tsitsiklis1986distributed,olfati2004consensus,talebi2019distributed,Guilherme2020ProcIEEE}. In \emph{average consensus}, agents repeatedly exchange information with their neighbors and seek to converge to the arithmetic average of their initial values. Classical distributed averaging schemes formulate this process through linear iterative updates and characterize conditions for convergence and the associated convergence rate \cite{xiao2004fast,olshevsky2009convergence}. Randomized communication protocols provide an alternative approach in which agents interact with randomly selected neighbors, giving rise to the class of randomized gossip algorithms \cite{Boyd2006TIT,dimakis2008geographic,dimakis2010gossip,zouzias2015randomized,Loizou2021TIT}. More recently, accelerated averaging methods based on lifted Markov chains and block-gossip mechanisms have been developed to improve convergence performance \cite{Jung2010TIT,liu2013analysis,Haddock2022TIT}.

From a broader stochastic and information-theoretic perspective, distributed averaging has also been studied under various forms of communication uncertainty. Moallemi and Van Roy \cite{moallemi2006consensus} characterized convergence of asynchronous distributed averaging, while Xiao et al.~\cite{xiao2007distributed} and Aysal and Barner \cite{aysal2010convergence} investigated average consensus under stochastic disturbances. Shi et al.~\cite{Guodong2015TIT} established almost-sure convergence and convergence-time results for consensus over random graph processes. Information-theoretic approaches have further investigated fundamental aspects of distributed averaging, including privacy-preserving averaging \cite{Fereydounian2023TIT}. These works provide important characterizations of averaging under stochastic or uncertain communication, but do not explicitly incorporate feedback-assisted retransmissions into the averaging dynamics.

In many networked systems, communication is inherently directional, resulting in directed and possibly unbalanced network topologies due, for example, to heterogeneous transmission power and communication ranges \cite{chen2023event}. In a directed network, information sent from agent $v_i$ to agent $v_j$ need not be available in the reverse direction. Consequently, standard linear consensus iterations generally do not preserve the arithmetic average. \emph{Ratio Consensus} addresses this issue by maintaining an auxiliary state together with the consensus state and using their ratio to compensate for asymmetric information flow \cite{dominguez2011distributed}. This approach achieves exact average consensus over strongly connected directed graphs under reliable communication \cite{dominguez2011distributed,cai2012average,dominguez2012distributed,priolo2014distributed}.

In practical wireless networks, however, communication links are unreliable and may introduce both information delays and packet losses. These impairments can prevent standard consensus schemes from preserving the exact average. For bounded but unknown time-varying delays, Hadjicostis and Charalambous \cite{hadjicostis2013average} proposed \emph{Robustified Ratio Consensus}, but their approach does not accommodate packet drops. Random graph processes with link failures have also been considered in \cite{Guodong2015TIT}; however, the resulting limiting consensus value need not equal the arithmetic average because the corresponding random update matrices do not necessarily preserve the network sum. To explicitly address packet losses in directed networks, Hadjicostis \etal \cite{hadjicostis2015robust} proposed \emph{Ratio Consensus via Running Sums} (\rcrs{}), in which nodes exchange running sums and use virtual nodes and links to represent lost information. Although \rcrs{} achieves exact average consensus despite packet losses, its running-sum variables grow with time, leading to increasing numerical dynamic-range requirements and, under fixed-width representations, potential quantization or overflow issues.

These limitations motivate considering unreliable communication together with the feedback mechanisms often used to improve packet-delivery reliability. From an information-theoretic perspective, distributed computation over unreliable channels involves a fundamental trade-off between communication reliability and computation time. Ayaso \etal \cite{ayaso2010information} established information-theoretic lower bounds for distributed function computation over noisy point-to-point channels, including linear-function computation over erasure channels, but did not consider iterative exact-average consensus. Fagnani and Zampieri \cite{fagnani2009average} directly studied average consensus under stochastic packet losses and characterized its mean-square performance and convergence rate. These works, however, do not explicitly incorporate feedback-assisted retransmissions into the consensus dynamics. In particular, they do not capture the coupled effects of retransmission-dependent reliability, retransmission-induced delays, and eventual packet drops under a finite retransmission limit. This leaves open the question of whether exact average consensus can be guaranteed over unreliable directed networks while explicitly accounting for feedback-assisted retransmissions, and what convergence and communication costs result from doing so.

This article builds on the preliminary conference studies in~\cite{makridis2023utilizing,makridis2023harnessing}, but develops a substantially more general stochastic and analytical framework. In particular, the present treatment is not limited to demonstrating the use of \texttt{ACK/NACK}-based retransmissions for consensus: it explicitly models retransmission-dependent decoding events, establishes an exact sample-path representation of the resulting mass evolution, and provides almost-sure as well as nonasymptotic convergence guarantees. More specifically, we make the following key contributions:

\begin{itemize}[leftmargin=*]
\item We introduce a new distributed consensus algorithm, referred to as \harqrc{},  for exact average consensus over strongly connected directed networks with unreliable communication, finite retransmission limits, retransmission-induced delays, and packet drops. \texttt{ACK/NACK} feedback allows each node to $(a)$ dynamically compute its out-degree using acknowledgment messages, $(b)$ determine successful packet delivery and recover the information mass associated with packets that are eventually dropped, and $(c)$ locally infer upper bounds on delays induced by the maximum retransmission limits.

\item We construct an augmented stochastic mass-flow representation for the delayed and dropped packets and establish its sample-path equivalence with the actual \harqrc{} iterations. Using this representation, we prove that every node converges almost surely to the exact arithmetic average of the initial values despite finite retransmission limits and permanent packet drops --- a scenario where existing average consensus algorithms fail or require strong assumptions, \eg no packet drops \cite{hadjicostis2013average}, or exchange of values that grow unbounded with time \cite{hadjicostis2015robust}.

\item Besides guaranteeing convergence to the exact average, the convergence speed of distributed averaging algorithms under network-topology constraints has also received considerable attention \cite{Jung2010TIT}. Beyond asymptotic convergence, we establish explicit performance guarantees, including a geometric bound on the expected consensus error, a high-probability error bound, and a finite upper bound on the expected $\epsilon$-consensus time. Numerical experiments further characterize the effect of the \harq{} reliability parameter and retransmission limit on the convergence--communication trade-off. 
\end{itemize}

The remainder of this paper is organized as follows. Section~\ref{sec:background} describes the notation used in this paper, introduces the network communication model and summarizes the main ideas behind ratio consensus-based algorithms and \harq{} error correction protocols. Section~\ref{sec:arq-consensus} formulates the probability of packet errors under the \harq{} protocol, and describes the proposed \harqratioconsensus{} algorithm for handling unreliable communication in the presence of packet errors. In Section~\ref{sec:convergence_analysis} we present the augmented digraph representation needed for the proof of convergence of the proposed method, while in Section~\ref{sec:numerical_evaluation} we provide numerical simulations under different channel condition scenaria. Finally, in Section~\ref{sec:conclusions} we provide concluding remarks and future directions.

\section{Background}\label{sec:background}
\subsection{Notation}
{ We denote the set of real (integer) numbers by $\mathbb{R}$ ($\mathbb{Z}$). The set of positive real numbers is denoted by $\posreal$. The set of nonnegative real numbers is denoted by $\nnreal$, and the set of nonnegative integer numbers by $\nnint$. The set of integer numbers  from $a\in\mathbb{Z}$ up to $b\in\mathbb{Z}$, $b>a$, is denoted by $\mathbb{Z}_{a}^{b}=\{a,a+1,\ldots,b-1,b\}$, \eg $\mathbb{Z}_{1}^{b}=\{1,2,\ldots,b\}$. The set of natural numbers is denoted by $\mathbb{Z}_{>0}$.} Vectors are denoted by lowercase letters whereas matrices are denoted by capital letters. The all-ones $n$-dimensional column vector is denoted by $\mathbf{1}_{n}$. The identity matrix and the zero matrix (of appropriate dimensions) are denoted by $I$ and $0$, respectively. A nonnegative matrix (with nonnegative elements) is denoted by $A\geq0$, while a positive matrix (with positive elements), is denoted by $A>0$.

For an event $\mathcal{C}$, $\mathbb{P}(\mathcal{C})$ denotes its probability, while $\mathbb{P}(\mathcal{C}\mid\mathcal{F})$ denotes its conditional probability given the information represented by the
$\sigma$-algebra $\mathcal{F}$. The expectation of a random variable $X$ is denoted by $\mathbb{E}\{X\}$, and
$\mathbb{E}\{X\mid\mathcal{F}\}$ denotes its conditional expectation given $\mathcal{F}$. The indicator of an
event $\mathcal{C}$ is denoted by $\mathbf{1}_{\{\mathcal{C}\}}$. Finally, ``\emph{a.s.}'' denotes almost-sure convergence or occurrence, \ie an event that holds with probability one.

\subsection{Graph Theory}
Consider a \emph{strongly connected}\footnote{A directed graph is \emph{strongly connected} if each node can be reached by every other node via a sequence of edges, respecting their orientation.} directed network captured by a graph $\set{G}=(\set{V}, \set{E})$, where $\set{V}=\{v_1, \cdots, v_n\}$ is the set of nodes (representing the $n$ agents) and $\set{E} \subseteq \set{V} \times \set{V}$ is the set of edges (representing the communication links between agents). The total number of edges in the network is denoted by $m=|\set{E}|$. A directed edge $\varepsilon_{ji} \triangleq (v_j, v_i) \in \set{E}$, where $v_j, v_i \in \set{V}$, indicates that node $v_j$ can receive information from node $v_i$, \ie $v_i \rightarrow v_j$. The nodes that transmit information to node $v_j$ directly are called in-neighbors of node $v_j$, and belong to the set $\set{N}_{j}^{-}=\{v_i \in \set{V} \mid \varepsilon_{ji} \in \set{E}\}$. The number of nodes in the in-neighborhood is called in-degree and it is represented by the cardinality of the set of in-neighbors, $d_{j}^{-} = |\set{N}_{j}^{-}|$. The nodes that receive information from node $v_j$ directly are called out-neighbors of node $v_j$, and belong to the set $\set{N}_{j}^{+}=\{v_l \in \set{V} \mid \varepsilon_{lj} \in \set{E}\}$. The number of nodes in the out-neighborhood is called out-degree and it is represented by the cardinality of the set of out-neighbors, $d_{j}^{+}= |\set{N}_{j}^{+}|$. Note that when self-loops are included, the number of in-going links of node $v_j$ is $(d_j^- {+} 1)$ and similarly the number of its out-going links is $(d_j^+ {+} 1)$.

\subsection{Problem Setup}
The problem of (distributed) average consensus involves a network where a number of nodes cooperate by exchanging information to compute the network-wide average of their initial values, for a certain quantity of interest. At each time step $k\geq0$, each node $v_j \in \set{V}$ maintains a state $x_j[k]$ and updates it at each iteration, based on the values received from its in-neighbors and its own value. In what follows, for simplicity of exposition, we assume that the state of each node $x_j[k]$ is a scalar, but in more general settings, it could be a vector. The goal of the nodes is to collaboratively compute the average of their initial values, hereinafter referred to as \emph{average consensus value} and denoted by
	\begin{align}\label{eq:ac_problem}
   		 \overline{x} \triangleq \dfrac{1}{n} \sum_{j=1}^{n} x_j[0].
	\end{align}
{ Towards this goal, we consider that the communication links enabling the information exchange between the nodes, are subject to packet errors. We assume that the channel conditions associated with different links and different packets are independent. Specifically, each directed link $\varepsilon_{lj} \in \set{E}$ is associated with a time-varying Bernoulli random variable $\alpha_{lj}[k]$ indicating whether a packet $\texttt{p}_{j}[k]$ transmitted at time $k$ by node $v_{j}$ to all of its out-neighbors including node $v_{l}$, is received erroneously at node $v_l$, \ie
\begin{align}\label{eq:packet_errors}
	\!\!\alpha_{lj}[k] = 
	\begin{cases}
		0, & \text{if~$\texttt{p}_{j}[k]$ erroneous $($with prob. $q_{lj}<1$$)$},\\
		1, & \text{otherwise $($with prob. $1-q_{lj}>0$$)$}.
	\end{cases}
\end{align}
For a given packet, however, decoding outcomes across successive \harq{} retransmissions need not be independent, since the receiver may combine information accumulated from previous attempts. Their conditional error probabilities are specified in Section~\ref{sec:packet_error_model}.

In the absence of global knowledge, nodes are required to execute an iterative distributed algorithm to eventually converge to the average consensus value, despite any occurred packet errors over the directed communication links.}

\subsection{Ratio Consensus over Reliable Directed Graphs}
In the \emph{Ratio Consensus} algorithm \cite{dominguez2011distributed}, at each time step $k$ each node $v_j$ maintains a state variable $x_j[k] \in \mathbb{R}$ and an auxiliary variable $y_j[k] \in \mathbb{R}_{+}$ that is used to compensate for the unbalanced information flow. The former is initialized at an arbitrary value $x_j[0]\in\mathbb{R}$ which represents the initial state/measurement of node $v_j$, while the latter is initialized at $y_j[0]=1$. In the iterative scheme of the ratio consensus algorithm, each node $v_{j}\in\set{V}$ receives a pre-weighted piece of information from its in-neighbors and updates its states (which are used to compute the ratio) according to:
\begin{subequations}\label{eq:ratio_consensus}
  \begin{align}
        x_j[k+1] &= p_{jj} x_j[k] + \sum_{v_i\in \inneighbor{j}} p_{ji} x_i[k],\\
        y_j[k+1] &= p_{jj} y_j[k] + \sum_{v_i\in \inneighbor{j}} p_{ji} y_i[k],\label{eq:auxiliary_y}\\
        z_j[k+1] &= x_j[k+1]/y_j[k+1],
   \end{align}
\end{subequations}
where the collection of weights $P=\{p_{ji}\} \in \nnreal^{n \times n}$ on links $\varepsilon_{ji} \in \set{E}$, forms a column-stochastic matrix.  Often, each node $v_j$ assigns the weight $p_{lj}$ as:
\begin{align}\label{eq:weights}
    p_{lj}=\begin{cases}
        1/(1 + \outdegree{j}), & v_l \in \outneighbor{j} \cup \{v_j\},\\
        0, & \text{otherwise},
    \end{cases}
\end{align}
which requires each node $v_j \in \set{V}$ to know its out-degree, and transmit $p_{lj} x_{j}[k]$ and $p_{lj} y_{j}[k]$ to its out-neighbors. Possible zero-valued entries in the nonnegative weighted column-stochastic adjacency matrix $P$, denote the absence of communication links (edges) between the corresponding nodes in the digraph.

\begin{rem}
The Ratio Consensus algorithm assumes that communication links within the network are perfectly reliable, and that each node $v_j$ is aware of its out-degree.
\end{rem}

\subsection{Robustified Ratio Consensus over Delayed Directed Graphs}\label{subsec:robustified_ratio_consensus}

In practice, packet transmissions are often delayed due to poor channel conditions and network congestion. To address this, the authors in \cite{hadjicostis2013average} proposed a protocol that handles arbitrary time-varying delays, ensuring asymptotic consensus to the exact average despite the delays. Let node $v_{j}$ at time step $k$ undergo an \emph{a priori} unknown delay, denoted by a bounded positive integer $\tau_{ji}[k] \leq \bar{\tau}_{ji}<\infty$. The maximum delay in the network is denoted by $\bar{\tau}\triangleq\max\{\bar{\tau}_{ji}\}$. Moreover, the own value of node $v_j$ is always instantly available without delay, \ie $\tau_{jj}[k]=0, \forall~k$. Based on this notation, the strategy proposed in \cite{hadjicostis2013average} requires each node to update its states at each iteration according to
\begin{align}\label{eq:robustified_ratio_consensus}
x_{j}[k+1] &= p_{jj} x_j[k] + \sum_{v_i \in \inneighbor{j}} \sum_{r=0}^{\bar{\tau}} p_{ji} x_{i}[k-r] \iota_{ji}[k-r],\nonumber\\
y_{j}[k+1] &= p_{jj} y_j[k] + \sum_{v_i \in \inneighbor{j}} \sum_{r=0}^{\bar{\tau}} p_{ji} y_{i}[k-r] \iota_{ji}[k-r],\nonumber\\
z_j[k+1] &= x_j[k+1]/y_j[k+1].
\end{align}
The indicator function, $\iota_{ji}[k-r]$, indicates whether the bounded delay $\tau_{ji}[k-r] \leq \bar{\tau}_{ji}$ on link $\varepsilon_{ji}$ at iteration $k-r$, equals $r$ (\ie the transmission on link $\varepsilon_{ji}$ at $k-r$ arrives at node $v_j$ at iteration $k$) and is defined as
\begin{align}
\iota_{ji}[k-r]= \begin{cases}1, & \text { if } \tau_{ji}[k-r]=r, \\ 0, & \text { otherwise. }\end{cases}
\end{align}
Here, it is important to note that a transmission on link $\varepsilon_{ji}$ undergoes an \emph{a~priori} unknown delay, for which node $v_j$ is unaware of; instead, node $v_{j}$ processes (delayed) packets as soon as they arrive successfully. Clearly, in the absence of delays, this strategy reduces to the Ratio Consensus algorithm in \cite{dominguez2011distributed} as described in \eqref{eq:ratio_consensus}.

\begin{rem}
The Robustified Ratio Consensus algorithm handles (possibly) delayed information exchange between nodes, neglecting any packet drops due to erroneous packets, while it assumes that each node $v_j$ is aware of its out-degree.
\end{rem}

\subsection{Ratio Consensus via Running Sums over Packet-Dropping Directed Graphs}\label{subsec:rc_rs}
To handle the loss of information packets, each node $v_{j} \in \set{V}$ executing the \rcrs{} algorithm \cite{hadjicostis2015robust}, maintains the variables $\sigma^{x}_{j}[k]$ and $\sigma^{y}_{j}[k]$ for the $x$ and $y$ mass, respectively. These variables, referred to as the \emph{running sum} values, contain the accumulated values of the state of node $v_{j}$ up to time $k$. The running sum values are transmitted to the out-neighbors $v_{l}\in\outneighbor{j}$ at each time step $k$. The transmission and processing of these variables (discussed shortly) allow nodes to converge to the average consensus value even when transmitting nodes are not notified of the successful delivery of their transmitted values.

Each node $v_{j}$ keeps track of the running sum values that (eventually) arrive from its in-neighbors by updating the variables $\chi_{ji}[k]$ and $\psi_{ji}[k]$. More specifically, each node $v_{j}$ initializes the aforementioned variables at $0$, \ie $\sigma^{x}_{j}[0]=0$, $\sigma^{y}_{j}[0]=0$, $\chi_{ji}[0]=0$, and $\psi_{ji}[0]=0$ for all $v_{i}\in\inneighbor{j}$. Essentially, at each time step $k$, node $v_{j} \in \set{V}$ transmits its running sum values $\sigma^{x}_{j}[k]+p_{lj}x_{j}[k]$ and $\sigma^{y}_{j}[k]+p_{lj}y_{j}[k]$ to its out-neighbors $v_{l}\in\outneighbor{j}$. Meanwhile, it receives $\sigma^{x}_{i}[k]+p_{ji}x_{i}[k]$ and $\sigma^{y}_{i}[k]+p_{ji}y_{i}[k]$ from all its in-neighbors $v_{i} \in \inneighbor{j}$ and updates the received masses as follows: 
\begin{align}
	\chi_{ji}[k+1] &= \theta_{ji}[k] \big(\sigma^{x}_{i}[k] + p_{ji}x_{i}[k] \big) + \big(1 - \theta_{ji}[k] \big) \chi_{ji}[k],\nonumber\\
	\psi_{ji}[k+1] &= \theta_{ji}[k] \big(\sigma^{y}_{i}[k] + p_{ji}y_{i}[k] \big) + \big(1 - \theta_{ji}[k] \big) \psi_{ji}[k],\nonumber
\end{align}
where $\theta_{ji}[k]=1$ if the masses $\sigma^{x}_{i}[k] + p_{ji}x_{i}[k]$ and $\sigma^{y}_{i}[k] + p_{ji}y_{i}[k]$, sent by node $v_{i}$ at time $k$, have been successfully received by node $v_{j}$; while $\theta_{ji}[k]=0$ otherwise. According to the updated mass variables $\chi_{ji}[k+1]$ and $\psi_{ji}[k+1]$, node $v_{j}$ updates its current state and auxiliary variables as 
\begin{align}
x_{j}[k+1] &= p_{jj} x_j[k] + \sum_{v_i \in \inneighbor{j}} \Big( \chi_{ji}[k+1] - \chi_{ji}[k] \Big),\nonumber\\
y_{j}[k+1] &= p_{jj} y_j[k] + \sum_{v_i \in \inneighbor{j}} \Big( \psi_{ji}[k+1] - \psi_{ji}[k] \Big),\nonumber\\
z_j[k+1] &= x_j[k+1]/y_j[k+1].
\end{align}
\begin{rem}
The \rcrs{} algorithm handles packets that are dropped during the information exchange between nodes, by having each node transmit its running sum values. Each node is assumed to be aware of its out-degree, while the values of running sums grow in an unbounded manner.
\end{rem}

\subsection{Error Correction Protocols}\label{subsec:ARQ}

In modern message transmission systems such as the Transmission Control Protocol (TCP), the High-Level Data Link protocol, and others, nodes utilize Automatic Repeat reQuest (\arq{}), protocols to maintain reliable packet transmissions over unreliable communication channels~\cite{lin1984automatic}, \cite[Chapter~6]{krouk2011modulation},\cite[Chapter~5]{leon2000communication}. \arq{}-based techniques have been recently applied in estimation \cite{huang2020real} to enhance the reliability and efficiency of data transmission in networked control systems, ensuring accurate state estimation even in the presence of communication errors. Although in these applications maintaining information integrity is crucial for system convergence and performance, one may also face the fundamental trade-off between the reliability and freshness of the information \cite{kosta2017age,huang2019retransmit}.

The \arq{} protocol employs error-detection codes, acknowledgment (\texttt{ACK}) or negative acknowledgment (\texttt{NACK}) messages, and retransmissions to ensure reliable data transmissions over error-prone channels. An acknowledgment is a feedback signal, often a single bit, sent by the receiver to notify the transmitter of successful or unsuccessful packet reception. In \arq{}, a packet is retransmitted after each \texttt{NACK} until it is successfully received, otherwise it is discarded. If the receiver fails to receive the packet after a predefined number of retransmission trials, the packet is dropped. Each retransmission of a packet can be seen as an information delay, since the receiver will get the intended packet successfully as soon as the transmitted packet arrives without errors. \arq{} can be enhanced by combining it with forward error correction (FEC) codes, forming the Hybrid \arq{} (\harq{}) protocol. In \harq{}, transmitters encode messages redundantly using FEC codes, allowing receivers to correct some errors in a packet. Standard \arq{} is then used to correct errors not corrected by FEC. This combination reduces the probability of packet errors in subsequent time slots, as some errors can be corrected using FEC, minimizing the need for extra retransmissions.

Here it is important to note that, \texttt{ACK/NACK} is sent over a narrowband error-free feedback channel with negligible probability of error in the reception of this small packet, which usually comprises one bit. However, this feedback channel cannot support regular data transmission due to the high data rate required and, hence, larger bandwidth and higher-order modulation. For this reason, the network topology over which the nodes exchange information data packets could still be assumed directed. 

A graphical representation of the \harq{} error control feedback mechanism is shown in Fig.~\!\ref{fig:arq_flowchart}. The flowchart blocks and arrows in red color denote the combined error correction techniques used in the \harq{}, on top of the classical \arq{} mechanism.
\begin{figure}[ht]
    \centering
    \includegraphics[scale=0.97]{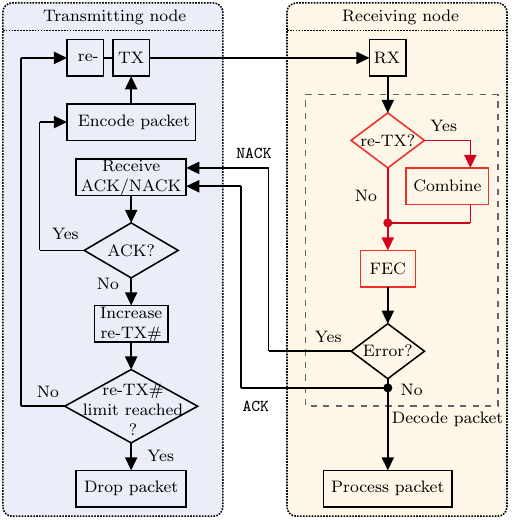}
    \caption{\harq{} error control feedback mechanism.}
    \label{fig:arq_flowchart}
\end{figure}

\section{Average Consensus via \harq{} Feedback Mechanisms}\label{sec:arq-consensus}

In this section, we present a novel distributed algorithm, referred to as the \emph{Hybrid Automatic Repeat reQuest Ratio Consensus} -- \harqrc{} algorithm, designed to achieve asymptotic convergence to the exact average consensus value in \eqref{eq:ac_problem}, in the presence of unreliable communication that is prone to packet errors. Each node communicates with its out-neighbors and in-neighbors by transmitting and receiving data packets, respectively, over a strongly connected directed graph $\set{G}=(\set{V},\set{E})$, employing the \harq{} protocol with a predetermined retransmission limit. Within this framework, at each consensus iteration, each node transmits its current state, while it may also retransmit some of its previous states, depending on the packets that arrived in error during previous transmissions. More specifically, each node assigns weights to its out-going links based on its out-degree as in formula \eqref{eq:weights}, and (re)transmits its states via broadcast, multicast, or unicast, depending on the acknowledgement feedback of the packets sent in the previous iteration. The out-degree of each node is initially acquired by broadcasting dummy packets and aggregating the resulting incoming \harq{} acknowledgement messages (\texttt{ACK/NACK}) from its out-neighbors. Counting these messages establishes that each node will acquire its out-degree for assigning the weights and multicasting its packets in the subsequent consensus iterations, even in the presence of packet errors.

\subsection{Modeling Packet Error Control with \harq{}}\label{sec:packet_error_model}
In this work, we assume that each node $v_{j}\in\set{V}$ employs a \harq{} protocol with a retransmission limit on its out-going links $\varepsilon_{lj} \in \set{E}$, denoted by $0\leq\bar{\tau}_{lj} < \infty$. This retransmission limit determines the maximum number of retransmissions selected before considering a packet as dropped. Thus, each node $v_{j}\in\set{V}$ should maintain $\bar{\tau}_{lj}$ buffers for possible retransmissions to its out-neighbor $v_{l}\in\outneighbor{j}$. Let $\texttt{p}_{j,h}[k]$ denote the packet generated and transmitted for the first time by node $v_{j}$ at time $h \in \set{H}_{k}$ where
\begin{align}\label{eq:Hk}
\mathcal H_k
\triangleq
\Big\{
k-r:\; r\in\mathbb Z_{\geq0},\;
0\leq r\leq \min\{k,\bar{\tau}_{lj}\}
\Big\},
\end{align}
and retransmitted for the $r$th time at time $k$. Moreover, let $f_{lj,h}[k]$ denote the acknowledgment message determined by node $v_{l}$ for the packet $\texttt{p}_{j,h}[k]$ it received at time $k$, to be defined as:
\begin{align}\label{eq:feedback}
		f_{lj,h}[k] \triangleq 
		\begin{cases}
			0, & \text{ if } \alpha_{lj,h}[k]=0~(\texttt{NACK}),\\
			1, & \text{ otherwise } (\texttt{ACK}),
		\end{cases}
\end{align}
where $\alpha_{lj,h}[k]\in\{0,1\}$ is a Bernoulli random variable that indicates the success of the $h$-th (re)transmission attempt of a message sent over the link $\varepsilon_{lj}\in\set{E}$ at time $k\in\nnint$. Each acknowledgement message $f_{lj,h}[k]$ is generated after decoding the corresponding packet during time slot $k$ and is returned to the transmitting node before the state-update phase of the same slot. Hence, it can be used in computing the state at time $k+1$, and remains available at the beginning of time slot $k+1$ for determining subsequent retransmissions. Hence, at time $k+1$, node $v_{j}$ is able to determine whether to retransmit previously erroneous packets based on the available acknowledgement messages it received from $v_{l}$, and the retransmission number of the packets. In general, for any packet that node $v_{j}$ generates and transmits to $v_l$ at time $h \in \set{H}_k$, the retransmission counter is updated as:
\begin{align}\label{eq:retransmission_counter}
   \rho_{lj,h}[k+1] \triangleq \begin{cases}
        \rho_{lj,h}[k] + 1, & \text{if}~f_{lj,h}[k]=0,\\
        0, & \text{if}~f_{lj,h}[k]=1,\\
        \rho_{lj,h}[k], & \text{otherwise,}
    \end{cases}
\end{align}
with $h\in\set{H}_{k}$ and $\rho_{lj,h}[k]=0$ for $k=h$. Based on the retransmission counters and the acknowledgement messages available at time $k$, node $v_{j}$ decides which packets to (re)transmit according to the condition:
\begin{equation}
\rho_{lj,h}[k]\leq\bar{\tau}_{lj} \quad\wedge\quad f_{lj,h}[k-1]=0,
\label{eq:retransmission_condition}
\end{equation}
for the packets generated at time $h \in \set{H}_{k-1}$. Thus, retransmissions performed during time slot $k$ are determined using the \texttt{ACK/NACK} information obtained from the preceding transmission attempt, whereas the feedback generated during time slot $k$ is available for the state update at the end of that slot and for retransmission decisions at time $k+1$. Recall the number of consecutive retransmissions $r$ of a packet generated and transmitted over the link $\varepsilon_{lj}$ at time $k$ until its successful reception by its destined node, results in delayed information. If the oldest packet generated and transmitted for the first time at time $k$ and retransmitted $\bar{\tau}_{lj}$ times failed to be received by $v_{l}$ at time $k+\bar{\tau}_{lj}$, then $v_{j}$ considers the packet as dropped, and stops retransmitting this packet. This is captured by
\begin{align}\label{eq:delayvalue}
\tau_{lj}[k] \triangleq 
\begin{cases}
\min\big\{ r \mid f_{lj,k}[k+r] = 1 \big\}, & r \in \mathbb{Z}_{0}^{\bar{\tau}_{lj}} ,   \\
\bar{\tau}_{lj}+1, & \text{otherwise}.
\end{cases}
\end{align}
Accordingly, $\tau_{lj}[k]=r\in\mathbb Z_{0}^{\bar{\tau}_{lj}}$ indicates successful reception after $r$ retransmissions, whereas the value $\tau_{lj}[k]=\bar{\tau}_{lj}+1$ is reserved to indicate
that the packet was not successfully decoded after the initial transmission and all $\bar{\tau}_{lj}$ allowed retransmissions. We therefore define the packet-drop indicator as
\begin{align}\label{eq:drop_indicator}
	\delta_{lj}[k] \triangleq 
	\begin{cases}
		1, & \!\!\!\text{if}~\tau_{lj}[k]=\bar{\tau}_{lj}+1 ,\\
		0, & \!\!\!\text{otherwise.}
	\end{cases}
\end{align}
We adopt a slotted communication model in which each time slot consists of a data-transmission phase, followed by a decoding and
feedback phase, and finally a state-update phase. Thus, the \texttt{ACK/NACK} associated with a packet (re)transmitted during time slot
$k$ is available to the transmitting node before the state update that generates $x_j[k+1]$ and $y_j[k+1]$. The same \texttt{ACK/NACK} is then
available at the beginning of time slot $k+1$ for determining whether the packet should be retransmitted. Consequently, if a packet generated at time $h$ fails its initial transmission and all
$\bar{\tau}_{lj}$ allowed retransmissions, its final \texttt{NACK} is obtained during time slot $h+\bar{\tau}_{lj}$, and the corresponding mass is
returned to the transmitting node in the state update at time $h+\bar{\tau}_{lj}+1$.

In what follows, we characterize the packet error probabilities for the standard \arq{} and the \harq{} protocol. Due to the absence of FEC in the standard \arq{} protocol, receiving nodes can only decode and correct the received packet by requesting retransmissions and expecting that the retransmitted packet will be successfully decoded without errors. Hence, the standard \arq{} protocol itself, cannot influence packet error probability, and at each iteration $k$, the packet error probability is independent of the packet retransmission number. For \harq{} protocols, however, the decoding outcome of a retransmission is generally dependent on the outcomes of previous attempts, since the receiver may combine information accumulated across transmissions. Several works (see, \eg \cite{HARQ:2001,HARQ:2010,ceran2019average}) have shown that the probability of packet error reduces exponentially with the number of retransmissions.

Before introducing the model on the reduction of probability of error with each retransmission, let us consider the following simplification. Assume that all links in the network have the same retransmission
limit, i.e., $\bar{\tau}_{lj}=\bar{\tau}$ for all $\varepsilon_{lj}\in\mathcal E$. For a successfully received packet, $\tau_{lj}[k]\in\mathbb{Z}_{0}^{\bar{\tau}}$, whereas $\tau_{lj}[k]=\bar{\tau}+1$ is reserved to indicate a packet drop, \ie $\delta_{lj}[k]=1$. Therefore, rather than assuming independent packet-error events across retransmissions, we model the conditional probability of decoding failure at the $r$th retransmission attempt. Specifically, for a packet generated by node $v_j$ and transmitted over link $\epsilon_{lj}$, the probability of a decoding failure is given by \cite{ceran2019average}: 
\begin{align}
\mathbb{P}\!\left(
\alpha_{lj,k}[k{+}r] \!=0 \mid \alpha_{lj,k}[k{+}s] \!= 0, s\in\mathbb{Z}_{0}^{r{-}1}
\right) = q_{lj}\lambda^r,
\label{eq:error_prob}
\end{align}
for all $k\geq0$ and $r \in \mathbb{Z}_{0}^{\bar{\tau}}$, where the conditioning is void for $r=0$. Here, $q_{lj}\in[0,1)$ denotes the packet-error probability of the initial transmission, while $\lambda \in (0,1]$ captures the improvement in decoding reliability obtained through \harq{} combining. Smaller values of $\lambda$ correspond to a faster reduction of the conditional packet-error probability across successive retransmissions. For $\lambda=1$, the conditional failure probability remains constant across attempts and the model reduces to standard \arq{}. Based on the packet error probability model in \eqref{eq:error_prob}, the probability that a packet initially transmitted over the link $\varepsilon_{ji}$ at time step $k$ is delayed by $r=0,1,\ldots,\bar{\tau}$ time steps or dropped, is defined as follows:
\begin{align}
    \mathbb{P}(\tau_{lj}[k]=0) &= 1 - q_{lj}, \tag*{(0 \texttt{NACK})} \nonumber\\
    \mathbb{P}(\tau_{lj}[k]=1) &= (1 - q_{lj}\lambda) q_{lj}, \tag*{(1 \texttt{NACK})} \nonumber\\
    &\;\;\vdots \nonumber\\
    \mathbb{P}(\tau_{lj}[k]=\bar{\tau}) &= 
    (1 - q_{lj}\lambda^{\bar{\tau}}) \textstyle{\prod\nolimits}_{r=0}^{\bar{\tau}-1} q_{lj} \lambda^r, \tag*{($\bar{\tau}$ \texttt{NACK}s)} \nonumber \\
    \mathbb{P}(\delta_{lj}[k]=1) &= \textstyle{\prod\nolimits}_{r=0}^{\bar{\tau}} q_{lj} \lambda^r. \tag*{(packet drop)} \nonumber
\end{align}
These probabilities can be obtained by the packet error transition model shown in Fig.~\!\ref{fig:error_prob_transition}.
\begin{figure}[ht]
    \centering
    \includegraphics[scale=0.85]{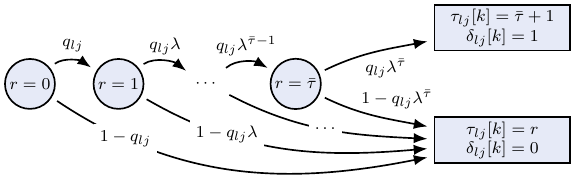}
    \caption{Packet error transition for a packet initially transmitted at iteration $k$ over the link $\varepsilon_{lj}$ with initial probability of error $q_{lj}$.}
    \label{fig:error_prob_transition}
\end{figure}

\subsection{\harqratioconsensus{} Algorithm -- \harqrc{}}

In this work, we assume that each transmitting node expects to receive back a separate acknowledgement for each (re)transmitted packet. In this scheme, each packet received by node $v_{l}$ is acknowledged individually, and the transmitting node $v_{j}$ decides whether to retransmit or not each individual packet. Note that this scheme requires transmissions and \harq{} feedback signals to be identified (matched to packets). 

Herein, we describe the \harqrc{} algorithm with individual acknowledgements, but one can adjust the algorithm for the collective feedback scheme\footnote{In the individual-acknowledgement scheme considered here, each transmitted packet is acknowledged separately by the receiving node, and the transmitting node decides whether to retransmit each packet based on its corresponding \texttt{ACK/NACK}. Thus, transmissions and HARQ feedback signals must be matched to individual packets. In the collective feedback scheme, all packets received during a time slot are acknowledged using a single \texttt{ACK/NACK} signal, where a \texttt{NACK} is sent if at least one of the received packets is erroneous.}. Specifically, each node $v_j$ executes the steps in Algorithm~\ref{alg:arq_based_consensus}, which can be summarized as follows:
\begin{itemize}[leftmargin=10pt]
\item \textbf{Input \& Initialization:} Node $v_{j}$ sets its initial state $x_j[0]\in\mathbb{R}$ which often corresponds to a local sensor reading, and its auxiliary variable at $y_j[0]=1$. 
\item \textbf{Out-degree acquisition:} The out-degree of node $v_j$ is acquired by broadcasting a dummy packet (or multiple ones to improve the out-degree acquisition accuracy) and summing the number of received acknowledgement messages from its out-neighbors. 
\item \textbf{At each time step $k\in\nnint$ node $v_{j}$:}
\begin{enumerate}[leftmargin=15pt,label=\textbf{(\alph*)}] 
	\item \textbf{Receives acknowledgement messages:} $f_{lj,h}[k-1]$ from each out-neighbor $v_{l}\in\outneighbor{j}$ for all $h\in\set{H}_{k-1}$ corresponding to the (re)transmissions of packets at the previous time slot.
	\item \textbf{Updates retransmission counters:} $\rho_{lj,h}[k]$ for each out-neighbor $v_{l}\in\outneighbor{j}$ and all $h\in\set{H}_{k}$ as in \eqref{eq:retransmission_counter}.
	\item \textbf{Transmits fresh packets:} $p_{lj}x_{j}[k]$ and $p_{lj}y_{j}[k]$ to its out-neighbors $v_{l}\in\outneighbor{j}$ via broadcast\footnote{Fresh packets can be broadcast since nodes $v_{j}\in\set{V}$ transmit the same value, \ie $p_{lj}x_{j}[k]$, to all their out-neighbors.}.
	\item \textbf{Retransmits previously erroneous packets:} $p_{lj}x_{j}[h]$ and $p_{lj}y_{j}[h]$ for $h\in\set{H}_{k-1}$ based on the received acknowledgement messages $f_{lj,h}[k-1]$ and the updated retransmission counters $\rho_{lj,h}[k]$ that fulfil the condition in \eqref{eq:retransmission_condition}.
	\item \textbf{Identifies dropped packets:} packets generated at time $k-\bar{\tau}$ for which $\delta_{lj}[k-\bar{\tau}]=1$, and returns their associated mass to the current state update.
	\item \textbf{Receives packets:} $p_{ji}x_{i}[h]$ and $p_{ji}y_{i}[h]$ from its in-neighbors $v_{i} \in \inneighbor{j}$, for $h\in\set{H}_{k}$ and for which $f_{ji,h}[k]=1$.
	\item \textbf{Sends acknowledgement messages:} $f_{ji,h}[k]$ to its in-neighbors $v_{i} \in \inneighbor{j}$ for $h\in\set{H}_{k}$, for requesting retransmission in case of packet error, \ie \texttt{NACK}, or to acknowledge the successful reception of the packet, \ie \texttt{ACK}. 
	\item \textbf{Updates states (generates fresh packets):} 
	$x_{j}[k+1]$ and $y_{j}[k+1]$ as in line (10) of Algorithm~\ref{alg:arq_based_consensus}, and then computes the ratio $z_{j}[k+1]=x_{j}[k+1]/y_{j}[k+1]$.  	
\end{enumerate}
\end{itemize}

\begin{algorithm}
\caption{\harqrc{} at node $v_j$}
\KwIn{$x_j[0]\in\mathbb{R}$, $\bar{\tau}<\infty$}
\KwInit{$y_j[0]=1$ and $\delta_{lj}[h]=0$ for $h<0$}
{ \textbf{Acquire out-degree} $\outdegree{j}\gets|\outneighbor{j}|$}

\For{$k\in\nnint$}{

    { \textbf{Update retransmission counters}}    
    $\rho_{lj,h}[k]$ for all $v_l\in\outneighbor{j}$ and
    $h\in\set{H}_k$ as in \eqref{eq:retransmission_counter}\;

    { \textbf{Transmit fresh packets}}
    $p_{lj}x_j[k]$ and $p_{lj}y_j[k]$
    to all $v_l\in\outneighbor{j}$\;

    { \textbf{Retransmit erroneous packets}}
    $p_{lj}x_j[h]$ and $p_{lj}y_j[h]$,
    for all $v_l\in\outneighbor{j}$ and
    $h\in\set{H}_{k-1}$ satisfying
    \eqref{eq:retransmission_condition}\;

    { \textbf{Receive acknowledgements}}
    $f_{lj,h}[k]$ for all packets
    (re)transmitted during slot $k$\;

    { \textbf{Identify dropped packets}}
    generated at $k-\bar{\tau}$
    satisfying $\delta_{lj}[k-\bar{\tau}]=1$\;

    { \textbf{Receive packets}}
    $p_{ji}x_i[h]$ and $p_{ji}y_i[h]$
    from all $v_i\in\inneighbor{j}$,
    for $h\in\set{H}_k$ satisfying $f_{ji,h}[k]=1$\;

    { \textbf{Send acknowledgements}}
    $f_{ji,h}[k]$ to all $v_i\in\inneighbor{j}$,
    for $h\in\set{H}_k$ as in \eqref{eq:feedback}\;

    { \textbf{Update states \& generate fresh packets}
    \vspace{-4pt}\begin{align*}
     x_j[k+1] &= p_{jj}x_j[k] + \sum_{v_i\in\inneighbor{j}}
     \sum_{h\in\mathcal H_k} f_{ji,h}[k]p_{ji}x_i[h] +\sum_{v_l\in\outneighbor{j}}
     \delta_{lj}[k-\bar{\tau}]
     p_{lj}x_j[k-\bar{\tau}],\\
     y_j[k+1] &= p_{jj}y_j[k]
      +\sum_{v_i\in\inneighbor{j}}
      \sum_{h\in\mathcal H_k}
      f_{ji,h}[k]p_{ji}y_i[h] + \sum_{v_l\in\outneighbor{j}}
      \delta_{lj}[k-\bar{\tau}]
      p_{lj}y_j[k-\bar{\tau}],
    \end{align*}\vspace{-7pt}}
    
    { \textbf{Ratio consensus estimate}}
    $z_j[k+1]\gets\frac{x_j[k+1]}{y_j[k+1]}$
}
\label{alg:arq_based_consensus}
\end{algorithm}

	\subsection{Example of \harqrc}
{ To get some intuition, consider a simple example where node $v_j$ (re)transmits the packets $\texttt{p}_{j,h}[k]$ for $h\in\set{H}_{k}$ to node $v_l$ at each time step $k$, as shown in Fig.~\!\ref{fig:multiple-packet-arq-feedback}. Note that, for simplicity of exposition we set $\bar{\tau}=2$, and we refer to the values that node $v_{j}$ transmits to its out-neighbor $v_{l}$, \ie $p_{lj}x_{j}[h]$, and $p_{lj} y_{j}[h]$, with the generic packet notation $\texttt{p}_{j,h}[k]$ we introduced in the previous section.}
\begin{figure}[ht]
    \centering
   \includegraphics[scale=0.94]{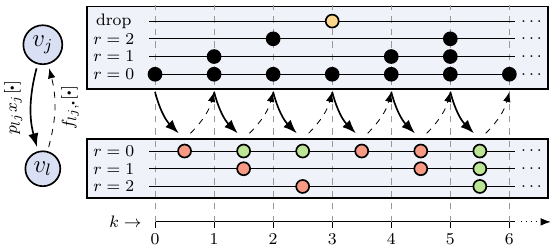}
    \caption{An acknowledgement message $f_{lj,h}[k]$ is sent from node $v_l$ for each packet it receives within each time slot. Black bullets denote a (re)transmission of packet $\texttt{p}_{j,h}[k]$; red and green bullets denote \texttt{ACK} and \texttt{NACK}, respectively; yellow bullets denote the iteration when node $v_{j}$ considers packets as dropped.}
    \label{fig:multiple-packet-arq-feedback}
\end{figure}

In this example, node $v_{l}$ sends back to node $v_{j}$ the acknowledgement messages $f_{lj,h}[k]$ for the packets it received (either successfuly or with errors) at time step $k$. Moreover, at each time $k$, node $v_{j}$ updates its retransmission counters as in \eqref{eq:retransmission_counter} to determine whether it should retransmit an erroneous packet. Based on this example, we present the variables for the information delays and packet drops, in Table~\ref{tab:simple_example}, and we briefly describe the \harq{} feedback mechanism for the first six iterations for a particular scenario of failed and successful packet (re)transmissions. To better explain the example, we consider slotted time where during each time slot, node $v_{j}$ receives the acknowledgement messages from previous transmissions and (re)transmits packets accordingly, and node $v_{l}$ acknowledges the packets that arrived within this time slot and sends back the corresponding acknowledgement messages.  
\begin{table}[h]\centering
\begin{tabular}{@{}rc|ccccccc@{}}\toprule
& & $k\!=\!0$ & $k\!=\!1$ & $k\!=\!2$ & $k\!=\!3$ & $k\!=\!4$ & $k\!=\!5$\\ \midrule
$\delta_{lj}[k]$ & & 1 & 0 & 0 & 0 & 0 & 0 \\ \midrule 
$\tau_{lj}[k]$ & & 3 & 0 & 0 & 2 & 1 & 0 \\
\bottomrule
\end{tabular}\vspace{5pt}
\caption{Packet drop indicator variables and information delay variables for the first six iterations of the transmissions over the link $\varepsilon_{ji}$ using the \harqrc{} algorithm.}
\label{tab:simple_example}
\end{table}

\noindent \textbf{\textit{Time slot $k=[0,1)$:}} Node $v_{j}$ sets $\rho_{lj,0}[0]=0$ and transmits a fresh packet $\texttt{p}_{j,0}[0]$. Node $v_{l}$ cannot process the packet it received since it was detected in error, \ie $f_{lj,0}[0]=0$.\\
\textbf{\textit{Time slot $k=[1,2)$:}} Node $v_{j}$ sets $\rho_{lj,1}[1]=0$, receives the acknowledgement message $f_{lj,0}[0]=0$, and updates the retransmission counter $\rho_{lj,0}[1]=1$. Then it transmits a fresh packet $\texttt{p}_{j,1}[1]$ and retransmits the previously failed packet $\texttt{p}_{j,0}[1]$. Node $v_{l}$ receives the packet $\texttt{p}_{j,1}[1]$ successfully, while the packet $\texttt{p}_{j,0}[1]$ is received again in error, and sends back $f_{lj,1}[1]=1$ and $f_{lj,0}[1]=0$, respectively.\\
\textbf{\textit{Time slot $k=[2,3)$:}} Node $v_{j}$ sets $\rho_{lj,2}[2]=0$, receives the acknowledgement messages $f_{lj,1}[1]=1$ and $f_{lj,0}[1]=0$, and updates its retransmission counters to $\rho_{lj,1}[2]=0$ and $\rho_{lj,0}[2]=2$, respectively. Based on these, it transmits a fresh packet $\texttt{p}_{j,2}[2]$, and retransmits for the last time the packet $\texttt{p}_{j,0}[2]$. Node $v_{l}$ receives successfully the fresh packet $\texttt{p}_{j,2}[2]$, while the previously failed packet arrives again in error, thus, it sends back $f_{lj,2}[2]=1$ and $f_{lj,0}[2]=0$, respectively. Since packet $p_{j,0}$ has now failed its initial transmission and both allowed retransmissions, its transmission outcome is $\tau_{lj}[0]=\bar{\tau}+1=3$, or equivalently $\delta_{lj}[0]=1$. Hence, its associated mass is returned to node $v_j$ in the state update producing $x_j[3]$ and $y_j[3]$. \\
\textbf{\textit{Time slot $k=[3,4)$:}} Packet $p_{j,0}$ has already been declared dropped at the end of the preceding time slot and is therefore no longer eligible for retransmission. Node $v_j$ sets $\rho_{lj,3}[3]=0$ and transmits the fresh packet $p_{j,3}[3]$. The packet is received in error by node $v_l$, which therefore returns the acknowledgement $f_{lj,3}[3]=0$.\\
\textbf{\textit{Time slot $k=[4,5)$:}}  Node $v_{j}$ sets $\rho_{lj,4}[4]=0$, receives the acknowledgement message $f_{lj,3}[3]=0$ and updates its retransmission counter to $\rho_{lj,3}[4]=1$. Then, it transmits a fresh packet $\texttt{p}_{j,4}[4]$, and retransmits the previously failed packet $\texttt{p}_{j,3}[4]$. Node $v_{l}$ receives both packets with errors and thus it sends back $f_{lj,4}[4]=0$ and $f_{lj,3}[4]=0$.\\ 
\textbf{\textit{Time slot $k=[5,6)$:}} Node $v_l$ receives all packets successfully without errors and therefore sends back $f_{lj,5}[5]=f_{lj,4}[5]=f_{lj,3}[5]=1$.

{ \section{Convergence Analysis}\label{sec:convergence_analysis}
In this section, we analyze the convergence properties of Algorithm~\ref{alg:arq_based_consensus}. We begin by introducing an augmented digraph representation that captures the information exchange induced by the \harqrc{} protocol under unreliable communication links. We then construct the corresponding augmented adjacency matrix and examine its structural properties. Finally, we present the main convergence theorem that constitutes the central theoretical result of this work.}

{ \subsection{Augmented Digraph Representation}\label{sec:aug_digraph}} 
For analytical simplicity, we assume a common retransmission
limit $\bar{\tau}$ for the \harq{} protocol over all communication
links, i.e., $\bar{\tau}_{ji}=\bar{\tau}$ for all $\varepsilon_{ji}\in\mathcal E$. Consequently, a successfully delivered packet may experience a delay of at most $\bar{\tau}$ iterations due to retransmissions, whereas the value $\tau_{ji}[k]=\bar{\tau}+1$ is reserved for the packet-drop event, as defined in~\eqref{eq:delayvalue}--\eqref{eq:drop_indicator}.
Hence, for each node $v_{j} \in \set{V}$, we add $\bar{\tau}$ extra virtual nodes, namely,  $\breve{v}_{j}^{(1)}, \breve{v}_{j}^{(2)}, \ldots, \breve{v}_{j}^{(\bar{\tau})}$, that temporarily hold and propagate the (delayed) information that is destined to arrive at the actual node $v_{j}$ after $r=1,\ldots,\bar{\tau}$ time steps. Furthermore, for each node $v_j\in\mathcal V$, we introduce
$\bar{\tau}$ additional virtual nodes,
$\check{v}^{(1)}_j,\check{v}^{(2)}_j,\ldots,
\check{v}^{(\bar{\tau})}_j$, which represent the return path of
the mass associated with packets that fail their initial
transmission and all $\bar{\tau}$ permitted retransmissions.
Such a packet is declared dropped after its final transmission
attempt during time slot $k+\bar{\tau}$, and its associated mass
is returned to the source-node state at time
$k+\bar{\tau}+1$. Finally, one can obtain the augmented digraph, $\set{\tilde{G}}=(\set{\tilde{V}}, \set{\tilde{E}})$, that models the network where nodes utilize \harq{} protocols, and consists of $\tilde{n} = n(2\bar{\tau}+1)$ nodes, where there are $n$ original nodes, $n\bar{\tau}$ virtual nodes due to the delays, and $n\bar{\tau}$ virtual nodes due to packet drops.

To illustrate the concept of digraph augmentation, consider the following example with two nodes exchanging information over an unreliable error-prone directed network, as shown in Fig.~\!\ref{fig:arq_graph}. In this example, both nodes $v_{i}$ and $v_{j}$ set the retransmission limit of their \harqrc{} algorithm to $\bar{\tau}=2$. Hence, the total number of nodes (actual and virtual) in the augmented digraph will be $\tilde{n} = n(2\bar{\tau}+1)=10$. In particular, for each actual node (shown in blue circle nodes) we add two virtual nodes (shown in yellow square nodes) corresponding to the delayed information due to packet retransmissions, and two virtual nodes (shown in orange square nodes) corresponding to the path followed when releasing the information back to its source node when a packet drop occurs. 
\begin{figure}[h]
    \centering
    \includegraphics[scale=0.85]{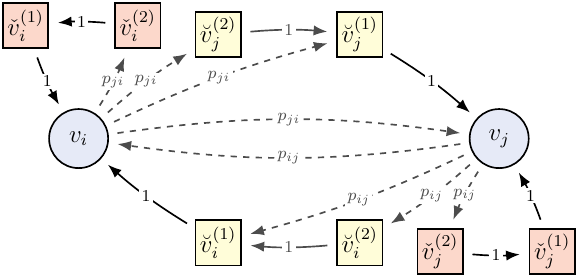}
    \caption{Two node augmented digraph corresponding to the exchange of information through the \harqrc{} algorithm. Yellow squares represent the virtual nodes due to retransmissions, while orange squares represent the virtual nodes due to packet drops.}
\label{fig:arq_graph}
\end{figure}

In essence, for a packet sent by $v_{i}$ over the link $\varepsilon_{ji} \in \set{E}$ and detected in error $r\in\mathbb{Z}_{0}^{\bar{\tau}}$ consecutive times ($r$ retransmissions) and then arrived successfully at its destined node $v_{j}$ without errors, the information in the augmented digraph will follow the path 
$$v_{i} \rightarrow \breve{v}_{j}^{(r)} \rightarrow \breve{v}_{j}^{(r-1)} \rightarrow \ldots \rightarrow \breve{v}_{j}^{(1)} \rightarrow v_{j}.$$ 
Conversely, if the packet fails its initial transmission and all
$\bar{\tau}$ permitted retransmissions, i.e.,
$\tau_{ji}[k]=\bar{\tau}+1$ and $\delta_{ji}[k]=1$, then its
associated mass follows the return path
$$v_{i} \rightarrow \check{v}_{i}^{(\bar{\tau})} \rightarrow  \check{v}_{i}^{(\bar{\tau}-1)} \rightarrow \ldots \rightarrow \check{v}_{i}^{(1)} \rightarrow v_{i}.$$

Based on the augmented digraph, we introduce the following sample-path representation of the information-mass evolution:
\begin{subequations}\label{eq:consensus_matrix_form}
\begin{align}
    \widetilde{\mathbf{x}}[k+1] = \Xi[k] \widetilde{\mathbf{x}}[k],\\
    \widetilde{\mathbf{y}}[k+1] = \Xi[k] \widetilde{\mathbf{y}}[k],
\end{align}
\end{subequations}
where $\widetilde{\mathbf{x}}[k]= \left( x[k]^{\top}, \breve{x}^{\top}[k],\check{x}^{\top}[k] \right)^{\top} \in \mathbb{R}^{\tilde{n}}$ (respectively $\widetilde{\mathbf{y}}[k]$) is a vector containing the variables of the actual and virtual nodes (propagating the delayed and dropped information) at time step $k$ as 
\begin{align*}
x[k] \! &= \! \big(x_1[k], \ldots, x_n[k]\big)^{\top} \!\! \in \mathbb{R}^{n}\\
\breve{x}[k] \! &= \! \Big( \breve{x}_1^{(1)}[k], \ldots, \breve{x}_n^{(1)}[k], \ldots, \breve{x}_1^{(\bar{\tau})}[k],\ldots,\breve{x}_n^{(\bar{\tau})}[k] \Big)^{\top} \!\! \in \mathbb{R}^{n\bar{\tau}},\\
\check{x}[k] \! &= \! \Big( \check{x}_1^{(1)}[k], \ldots, \check{x}_n^{(1)}[k], \ldots, \check{x}_1^{(\bar{\tau})}[k],\ldots,\check{x}_n^{(\bar{\tau})}[k] \Big)^{\top} \!\! \in \mathbb{R}^{n\bar{\tau}}.
\end{align*}
We emphasize that $\Xi[k]$ is introduced solely as an analytical sample-path representation. In particular, its construction may depend on the eventual \harq{} outcome $\tau_{lj}[k]$ of packets generated at time $k$ and therefore does not constitute a matrix that must be constructed or known by the nodes online. The augmented states are initialized as $\widetilde{x}_j[0]=x_j[0]$ and $\widetilde{y}_j[0]=1$ for all $v_j\in\mathcal V$, while
\begin{equation}
\widetilde{x}_\ell[0] = \widetilde{y}_\ell[0] = 0,~ \forall\, v_\ell\in \widetilde{\mathcal V}\setminus\mathcal V.
\label{eq:augmented_initialization}
\end{equation}
Matrix $\Xi[k] \in \nnreal^{\tilde{n} \times \tilde{n}}$ is a nonnegative matrix associated with the augmented digraph at iteration $k$:
\begin{align}\label{eq:Xi}
\setlength\arraycolsep{2.2pt}
\Xi[k] \triangleq \left(\begin{array}{ccccccccc}
P^{(0)}[k] & I & 0 & \cdots & 0 & I & 0 & \cdots & 0 \\
P^{(1)}[k] & 0 & I & \cdots & 0 & 0 & 0 & \cdots & 0 \\
\vdots & \vdots & \vdots & \ddots & \vdots & \vdots & \vdots & \ddots & \vdots\\
P^{(\bar{\tau}-1)}[k] & 0 & 0 & \cdots & I & 0 & 0 & \cdots & 0 \\
P^{(\bar{\tau})}[k] & 0 & 0 & \cdots & 0 & 0 & 0 & \cdots & 0 \\
0 & 0 & 0 & \cdots & 0 & 0 & I & \cdots & 0 \\
\vdots & \vdots & \vdots & \ddots & \vdots & \vdots & \vdots & \ddots & I \\
D^{(\bar{\tau})}[k] & 0 & 0 & \cdots & 0 & 0 & 0 & \cdots & 0
\end{array}\right),
\end{align}
where each element of $P^{(r)}[k]$, $r\in\mathbb{Z}_{0}^{\bar{\tau}}$, is determined as follows. For $r=0$,
\begin{equation}
\big[P^{(0)}[k]\big]_{lj}
\triangleq
\begin{cases}
p_{jj}, &
l=j,\\
p_{lj}, &
l\neq j,\;
\varepsilon_{lj}\in\mathcal E,\;
\tau_{lj}[k]=0,\\
0, & \text{otherwise},
\end{cases}
\label{eq:P0_definition}
\end{equation}
while for $r=1, \ldots, \tau$,
\begin{equation}
\big[P^{(r)}[k]\big]_{lj}
\triangleq
\begin{cases}
p_{lj}, &
\varepsilon_{lj}\in\mathcal E,\;
\tau_{lj}[k]=r,\\
0, & \text{otherwise}.
\end{cases}
\label{eq:Pr_definition}
\end{equation}
Thus, the diagonal entries of $P^{(0)}[k]$ always contain the
delay-free self-loop weights $p_{jj}$, whereas its off-diagonal
entries contain the masses that are successfully received without
retransmission. For $r\in\mathbb Z_{1}^{\bar{\tau}}$,
$P^{(r)}[k]$ contains the masses that are successfully received
after exactly $r$ retransmissions.

Moreover, $D^{(\bar{\tau})}[k]\in\mathbb{R}_{\geq 0}^{n\times n}$
is a diagonal matrix that routes the mass associated with packets
that are declared dropped after their final permitted transmission
attempt. Its $j$th diagonal entry is given by
\begin{equation}
\label{eq:D}
D^{(\bar{\tau})}_{jj}[k]
=
\sum_{v_l\in\outneighbor{j}}
\delta_{lj}[k]p_{lj},
\qquad
\forall\,v_j\in\mathcal V.
\end{equation}
Thus, whenever a packet generated by node $v_j$ at time $k$ and
intended for an out-neighbor $v_l\in\outneighbor{j}$ satisfies
\begin{equation}
\tau_{lj}[k]=\bar{\tau}+1,
\qquad\text{equivalently}\qquad
\delta_{lj}[k]=1,
\end{equation}
the corresponding mass $p_{lj}x_j[k]$ (and, analogously,
$p_{lj}y_j[k]$) is routed to the virtual drop state
$\check v_j^{(\bar{\tau})}$. The accumulated dropped mass then
propagates through the unit-weight virtual links
\begin{equation}
\check v_j^{(\bar{\tau})}
\rightarrow
\check v_j^{(\bar{\tau}-1)}
\rightarrow
\cdots
\rightarrow
\check v_j^{(1)}
\rightarrow
v_j,
\end{equation}
and is eventually returned to its source node $v_j$.
According to the slotted timing convention introduced in
Section~III-A, a packet generated at time $k$ whose final permitted
attempt fails during time slot $k+\bar{\tau}$ is returned to the
source-node state at time $k+\bar{\tau}+1$. Hence, the
$\bar{\tau}$ virtual nodes in the above drop chain exactly represent
the timing of the packet-drop recovery mechanism.

\begin{prop}[Sample-path equivalence]
\label{prop:augmented_equivalence}
Fix an arbitrary realization of the \texttt{ACK/NACK} outcomes generated
during the execution of Algorithm~1. Let
$\widetilde{\mathbf{x}}[k]$ and
$\widetilde{\mathbf{y}}[k]$ evolve according to
\eqref{eq:consensus_matrix_form},
with the initial conditions $\widetilde{x}_j[0]=x_j[0]$,
and  $\widetilde{y}_j[0]=y_j[0]=1$, $v_j\in\mathcal V$,
and zero initial states at all virtual nodes. Then, for every
$k\in\mathbb Z_{\geq0}$, the components of
$\widetilde{\mathbf{x}}[k]$ and
$\widetilde{\mathbf{y}}[k]$ associated with the actual nodes
coincide exactly with the states $\mathbf{x}[k]$ and
$\mathbf{y}[k]$ generated by Algorithm~1.
\end{prop}

\begin{proof}
It is sufficient to verify that, for every packet generated at
time $h$ and every communication link
$\varepsilon_{lj}\in\mathcal E$, the augmented dynamics deliver
or return the associated masses $p_{lj}x_j[h]$ and
$p_{lj}y_j[h]$ to the same actual node and at the same iteration
as Algorithm~1. Three mutually exclusive cases can occur.

First, suppose that the packet is successfully received during
its initial transmission, i.e., $\tau_{lj}[h]=0$. By~\eqref{eq:P0_definition}, the corresponding weight $p_{lj}$ belongs to $P^{(0)}[h]$. Hence, the masses
$p_{lj}x_j[h]$ and $p_{lj}y_j[h]$ are transferred directly to
the actual receiving node $v_l$ in the update from time $h$ to
time $h+1$. This is exactly the contribution generated by
$f_{lj,h}[h]=1$ in Algorithm~1.

Second, suppose that the packet is successfully received after
$r$ retransmissions, where $r\in\mathbb Z_{1}^{\bar{\tau}}$, i.e., $\tau_{lj}[h]=r$.
By~\eqref{eq:Pr_definition}, the corresponding masses enter the
$r$th delay virtual node associated with $v_l$ in the update from
$h$ to $h+1$. They subsequently propagate through the unit-weight
virtual links according to
$v_j
\rightarrow
\breve{v}_l^{(r)}
\rightarrow
\breve{v}_l^{(r-1)}
\rightarrow\cdots\rightarrow
\breve{v}_l^{(1)}
\rightarrow
v_l$. Consequently, the masses reach the actual receiving node $v_l$
in the state at time $h+r+1$. In Algorithm~1, the $r$th
retransmission occurs during time slot $h+r$, for which
$f_{lj,h}[h+r]=1$, and the same masses are incorporated into the
state update producing $x_l[h+r+1]$ and $y_l[h+r+1]$.
Therefore, the two representations coincide.

Finally, suppose that the packet is not successfully decoded
during its initial transmission or any of the
$\bar{\tau}$ allowed retransmissions. According to
\eqref{eq:drop_indicator}, $\tau_{lj}[h]=\bar{\tau}+1$ and $\delta_{lj}[h]=1$. The corresponding weight $p_{lj}$ is then included in
$D^{(\bar{\tau})}[h]$, and the associated mass follows the
packet-drop chain
$v_j
\rightarrow
\check{v}_j^{(\bar{\tau})}
\rightarrow
\check{v}_j^{(\bar{\tau}-1)}
\rightarrow\cdots\rightarrow
\check{v}_j^{(1)}
\rightarrow
v_j$.
Hence, the mass generated at time $h$ is returned to the actual
source node $v_j$ in the state at time $h+\bar{\tau}+1$.
Under the slotted timing convention of Section~\ref{sec:packet_error_model}, the final
NACK is obtained during time slot $h+\bar{\tau}$, and Algorithm~\ref{alg:arq_based_consensus}
returns the same mass through the term $\delta_{lj}[h]p_{lj}x_j[h]$
(and analogously for $y_j[h]$) in the state update producing
$x_j[h+\bar{\tau}+1]$ and $y_j[h+\bar{\tau}+1]$.

The delay-free self-mass $p_{jj}x_j[k]$ and
$p_{jj}y_j[k]$ is represented by the diagonal entries
$\big[P^{(0)}[k]\big]_{jj}=p_{jj}$ and therefore also coincides with
the corresponding terms in Algorithm~\ref{alg:arq_based_consensus}.

Since the initial actual-node states coincide and the three cases
above exhaust all possible outcomes of every transmitted mass,
each individual mass contribution is delivered to, or returned to,
the same actual node at the same iteration in both representations.
Therefore, by linearity of the updates and induction on $k$, the
actual-node components of $\widetilde{\mathbf{x}}[k]$ and
$\widetilde{\mathbf{y}}[k]$ coincide with
$\mathbf{x}[k]$ and $\mathbf{y}[k]$, respectively, for every
$k\in\mathbb Z_{\geq0}$.
\end{proof}

\subsection{Properties of Augmented Matrix}
Having established the sample-path equivalence between Algorithm~\ref{alg:arq_based_consensus} and the augmented dynamics, we next characterize the structural properties of $\Xi[k]$ that are used in the convergence analysis. Matrix $\Xi[k]$ maintains column-stochasticity at all $k\in\nnint$ and denotes a particular instance from the set of all possible instances of realized delays (induced from packet retransmissions), and packet drops. All these possible instances are in the set $\set{X}$, which consists of a finite number of matrices of identical dimensions $\tilde{n}\times \tilde{n}$, where each matrix in $\set{X}$ corresponds to a distinct instantiation of the packet drop indicator variable $\delta_{ji}[k]$ and the delay on the links $\tau_{ji}[k]$. Moreover, every positive entry of $\Xi[k]$ is lower bounded
by $\xi\triangleq\min_{v_j\in\mathcal V}{\big(1+\outdegree{j}\big)^{-1}}>0$. Indeed, the positive entries originating from the matrices $P^{(r)}[k]$ are equal to some $p_{lj}\geq\xi$, the virtual shift entries are equal to one, while each positive diagonal entry of $D^{(\bar\tau)}[k]$ is a sum of one or more weights $p_{lj}$ and is therefore also at least $\xi$. These properties can be exploited for the proof of convergence of the \harqrc{} algorithm. 

\subsection{Main Convergence Result}
\begin{assumption}[Uniform conditional transmission reliability]
\label{ass:conditional_success}
For every communication link $\varepsilon_{ij}\in\mathcal E$,
every packet generated at time $h\in\mathbb Z_{\geq 0}$, and every
admissible transmission attempt $r\in\mathbb Z_{0}^{\bar{\tau}}$,
there exists a constant $p_{\min}>0$ such that
\begin{equation}
\mathbb{P} \big( \alpha_{ij,h}[h+r]=1 \mid \mathcal F_{ij,r}[h]
\big) \geq p_{\min},
\label{eq:conditional_success_assumption}
\end{equation}
where $\mathcal F_{ij,r}[h]$ denotes the information available
immediately prior to the $r$th transmission attempt of the packet
generated at time $h$ over link $\varepsilon_{ij}$.
\end{assumption}
Assumption~\ref{ass:conditional_success} does not require the
decoding outcomes of successive retransmissions of the same packet
to be independent. It only requires the conditional probability of
successful decoding to be uniformly bounded away from zero over all
admissible transmission histories. Under the \harq{} model in
\eqref{eq:error_prob}, the conditional success probability at retransmission index $r$ is
\begin{equation}
1-q_{ij}\lambda^r \geq 1-q_{ij} \geq 1-q_{\max},
\end{equation}
where $q_{\max} \triangleq \max_{\varepsilon_{ij}\in\mathcal E}q_{ij}<1$. Hence, Assumption~\ref{ass:conditional_success} is satisfied with
$p_{\min}=1-q_{\max}$.

Recall that the objective of every actual node is to recover the
network-wide average of the initial values,
\begin{equation}
z^\star = \frac{ \sum_{v_\ell\in\widetilde{\mathcal V}}\widetilde{x}_\ell[0] }{
\sum_{v_\ell\in\widetilde{\mathcal V}}\widetilde{y}_\ell[0] } = \frac{1}{n}
\sum_{v_\ell\in\mathcal V}x_\ell[0]
= \overline{x}.
\label{eq:average_augmented}
\end{equation}
Since $y_j[0]=1$ and every actual node has the strictly positive
self-loop weight $p_{jj}>0$, the augmented dynamics imply
\begin{equation}
y_j[k+1] \geq p_{jj}y_j[k] >0, \quad \forall\,v_j\in\mathcal V,\; k\in\mathbb Z_{\geq0}.
\label{eq:positive_actual_y_main}
\end{equation}
Consequently, the ratio
\begin{equation}
z_j[k] \triangleq \frac{x_j[k]}{y_j[k]}
\end{equation}
is well defined at every actual node and at every iteration. Under Assumption~\ref{ass:conditional_success}, the following theorem establishes that these ratios converge almost surely to the exact average in~\eqref{eq:average_augmented}.
\begin{thm}
\label{theorem:arq_consensus}
Consider a strongly connected digraph $\mathcal G=(\mathcal V,\mathcal E)$, where each node $v_j\in\mathcal V$ is initialized with $x_j[0]\in\mathbb R$ and $y_j[0]=1$. Suppose that the communication links are unreliable and satisfy Assumption~\ref{ass:conditional_success}, and let $z_j[k]$, $v_j\in\mathcal V$, denote the ratios generated by Algorithm~\ref{alg:arq_based_consensus}. Then,
\begin{equation}
\mathbb{P}\!\left( \lim_{k\rightarrow\infty}z_j[k] = \frac{1}{n} \sum_{v_i\in\mathcal V}x_i[0], \quad \forall\,v_j\in\mathcal V \right) =1.
\label{eq:main_convergence}
\end{equation}
\end{thm}

The proof of Theorem~\ref{theorem:arq_consensus} is provided in Appendix~\ref{appendix:theorem}. The main argument establishes that the ratios of all active augmented nodes remain within a nonexpanding interval, while communication blocks that uniformly contract this interval occur infinitely often almost surely under
Assumption~\ref{ass:conditional_success}. Conservation of the augmented $x$- and $y$-masses then identifies the common limiting ratio with the exact average $\overline{x}$.

\section{Numerical Evaluation}\label{sec:numerical_evaluation}
Consider the strongly connected directed network $\set{G}=(\set{V},\set{E})$ shown in Fig.~\!\ref{fig:graph_numerical_example} consisting of five nodes, \ie $n=|\set{V}|=5$, with each node $v_j$ choosing the weights on its out-going links (including its self-loop link) as in \eqref{eq:weights}. Each node has immediate access to its current states, although the self-loop links are not shown for simplicity.
\begin{figure}[h]
\begin{center}
    \begin{minipage}{.47\textwidth}
    \centering
    \includegraphics{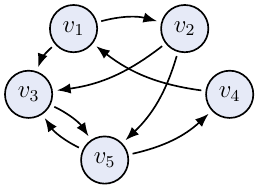}
    \caption{Five node strongly-connected digraph.}
\label{fig:graph_numerical_example}
\end{minipage}
\end{center}
\end{figure}

The variables of each node $x_j$, $y_j$ are updated at each time step using Algorithm~\ref{alg:arq_based_consensus} with initial values ${x}[0]=(x_1[0] \ldots x_5[0])^{\top}=(4~5~6~3~2)^{\top}$, and ${y}[0]=(y_1[0] \ldots y_5[0])^{\top}=(1~1~1~1~1)^{\top}$, respectively. In the examples that follow we consider that each packet initially transmitted over any link $\varepsilon_{ji}\in\set{E}$ may arrive in error with initial probability of packet error $q=q_{ji}$. Similarly, without loss of generality, we assume that all nodes $v_{i}\in\set{V}$ utilize an \harqrc{} algorithm of which the retransmission limit is set to $\bar{\tau}_{ji} = \bar{\tau}$ for all $v_{j} \in \outneighbor{i}$.

First, to gain insight into the operation of the \harqrc{} algorithm in terms of the state updates, consider the aforementioned network where each packet follows the error transition in \eqref{eq:error_prob} with initial probability of error $q=0.6$. In this example, we assume that each node is employed with the \harqrc{} algorithm with retransmission limit set to $\bar{\tau}=2$. To keep the example simple, consider that nodes acknowledge the received packets collectively, that is, they send an \texttt{ACK} if all packets arrived without errors, while they send a \texttt{NACK} if there exists at least one packet that arrived with errors. Fig.~\!\ref{fig:arq_consensus_transmissions_node1} shows the evolution of the variables and states at node $v_1$, which receives packets only from node $v_{4}$, during the first time steps of the execution of the \harqrc{} algorithm. 
\begin{figure}[]
\centering
\includegraphics[scale=0.85]{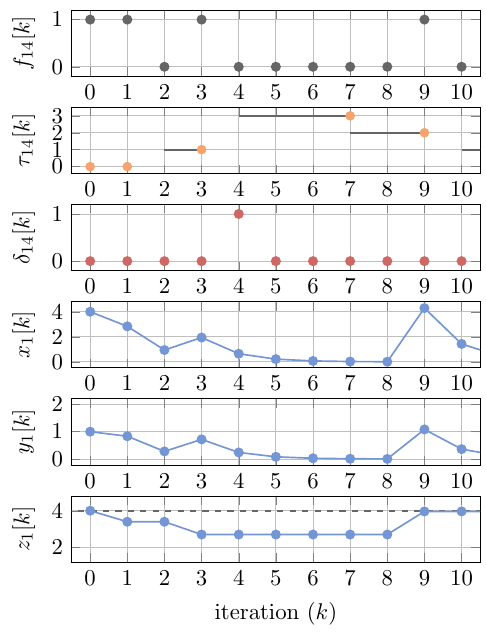}
\caption{Packet transmissions from node $v_4$ to node $v_1$ using the \harqrc{} algorithm and $q_{ji}=0.6$, $\bar{\tau}=2$.}
\label{fig:arq_consensus_transmissions_node1}
\end{figure}
As shown in Fig.~\!\ref{fig:arq_consensus_transmissions_node1}, the state variables $x_{1}[k]$ and $y_{1}[k]$ (and hence the ratio $z_1[k]$) are driven by the reception of packets $p_{14}x_{4}[k]$ that arrive without errors (when $f_{14}[k]=1$) from node $v_{4}$, and its own weighted state $p_{11}x_{1}[k]$ which is delay-free at all $k\geq 0$. Note that, the packets (re)transmitted between time steps $4$-$7$ are erroneous, and since the retransmission limit is exceeded ($\tau_{14}[4]=3$), the corresponding packets will be dropped and released back to node $v_{4}$ at time step $k=7$.

Next, we evaluate the performance of our proposed algorithm with $\lambda=1$, which we call, for simplicity, \arqrc{}, and \harqrc{} with $\lambda=0.5$. Here, $\lambda$ denotes the improvement in decoding reliability obtained through \harq{} combining, while $q$ denotes the packet-error probability of the initial transmission. We consider three different initial packet-error probabilities, \ie, $q=\{0.1,0.5,0.9\}$, and two different retransmission limits, \ie, $\bar{\tau}=\{1,3\}$. In Fig.~\!\ref{fig:arq_ratio_taubar1}-\ref{fig:harq_ratio_taubar3} we present the ratio $z_{j}[k]$ of each node $v_{j}\in\set{V}$ using the \arqrc{} and \harqrc{} algorithms, for $\bar{\tau}=1$ and $\bar{\tau}=3$. For relatively good channel conditions, \ie $q=0.1$, the convergence speed of the nodes is similar for both algorithms and both retransmission limits since most of the packets will arrive at their destined nodes successfully without errors. For $q=0.5$, we can observe a slightly better performance using the \arqrc{} with retransmission limit set to $\bar{\tau}=3$ rather than $\bar{\tau}=1$. Now, focusing on the case where $q=0.9$, it is more obvious that when setting a higher retransmission limit, \ie $\bar{\tau}=3$, for both algorithms, the convergence to the exact average consensus value is faster. The convergence speed-up becomes more apparent for the \harqrc{} when $\bar{\tau}=3$. This suggests that, when the channel conditions are bad, setting an appropriately chosen higher retransmission limit would allow the nodes to achieve faster convergence to the average consensus value. 

\begin{figure*}[h]
\centering
\includegraphics[width=\linewidth]{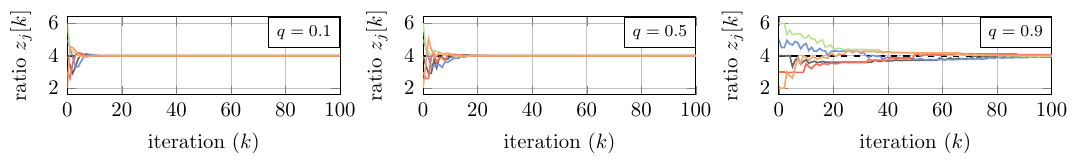}
\vspace{-4pt}
\caption{Ratio $z_j[k]$ at each node $v_j$ using \arqrc{} ($\lambda=1$) with retransmission limit $\bar{\tau}=1$.}
\vspace{-10pt}
\label{fig:arq_ratio_taubar1}
\end{figure*}
\begin{figure*}[h]
\centering
\includegraphics[width=\linewidth]{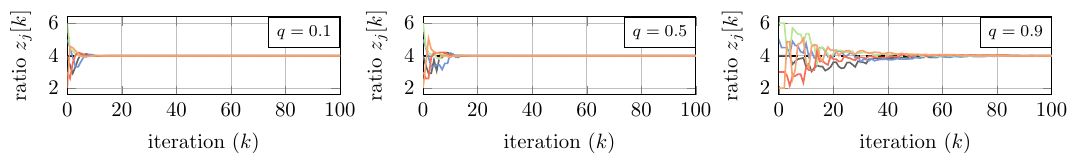}
\vspace{-4pt}
\caption{Ratio $z_j[k]$ at each node $v_j$ using \arqrc{} ($\lambda=1$) with retransmission limit $\bar{\tau}=3$.}
\vspace{-10pt}
\label{fig:arq_ratio_taubar3}
\end{figure*}
\begin{figure*}[h]
\centering
\includegraphics[width=\linewidth]{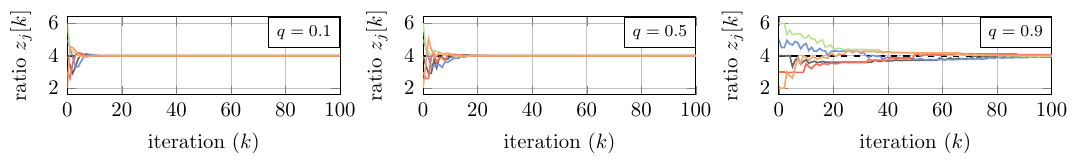}
\vspace{-4pt}
\caption{Ratio $z_j[k]$ at each node $v_j$ using \harqrc{} ($\lambda=0.5$) with retransmission limit $\bar{\tau}=1$.}
\vspace{-10pt}
\label{fig:harq_ratio_taubar1}
\end{figure*}
\begin{figure*}[h]
\centering
\includegraphics[width=\linewidth]{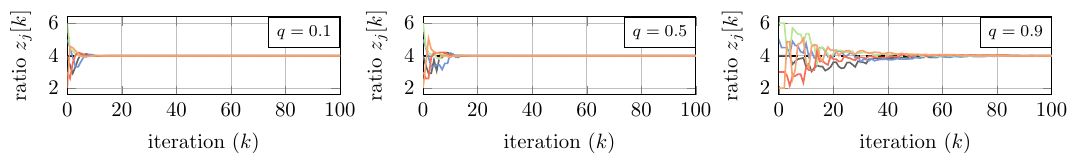}
\vspace{-4pt}
\caption{Ratio $z_j[k]$ at each node $v_j$ using \harqrc{} ($\lambda=0.5$) with retransmission limit $\bar{\tau}=3$.}
\vspace{-10pt}
\label{fig:harq_ratio_taubar3}
\end{figure*}

\subsection{Mean Residual vs. Retransmission Limit}
To quantify the performance of the proposed algorithm with respect to different retransmission limits $\bar{\tau}$, we consider three different channel conditions captured by different initial probability of error, \ie $q_{ji}=q\in\{ 0.1, 0.5, 0.9\}$. In particular, to obtain an overall assessment of convergence over all time steps of the \harqrc{} algorithm, we record the mean residual { over $S=100$ simulations of $K=200$ total iterations each, \ie 
\begin{align}
	\tilde{e} = \frac{1}{S} \sum_{s=1}^{S}  \frac{1}{K}\sum_{k=0}^{K-1}  \lVert z[k]-\overline{x}\mathbf{1}_{n} \rVert.
\end{align}%
The mean residual $\tilde{e}$ is obtained and presented for different retransmission limits $\bar{\tau}$, as shown in Fig.~\!\ref{fig:arq_mean_residual}.}
\begin{figure}[h]
\centering
\includegraphics[scale=0.8]{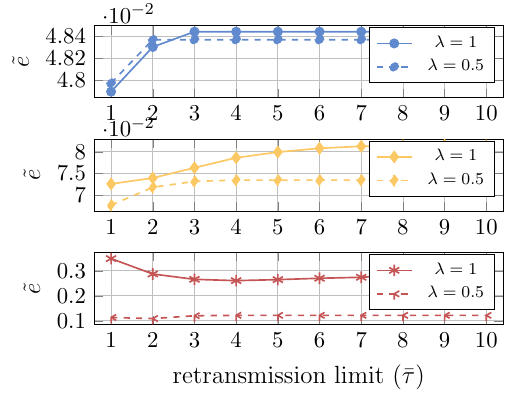}
\caption{Mean residual of \arqrc{} ($\lambda=1$) and \harqrc{} ($\lambda=0.5$), over different retransmission limits $\bar{\tau}$, for different initial packet error probabilities, \ie top: $q=0.1$, middle: $q=0.5$, and bottom: $q=0.9$.}
\label{fig:arq_mean_residual}
\end{figure}
This figure suggests that when channel conditions are good, \ie $q=0.1$, then it is preferable to set a lower retransmission limit, since, according to \eqref{eq:error_prob}, most of the packets arrive at the receiving node successfully without many retransmission trials, for both $\lambda=0.5$ and $\lambda=1$. Here it is worth mentioning that, for $\bar{\tau}\geq 3$, there is no improvement in terms of convergence speed. Similar behavior is observed for $q=0.5$, although the saturation in terms of convergence speed occurs for higher retransmission limits, \ie around $\bar{\tau}\geq 7$. In contrast, for bad channel conditions, \ie $q=0.9$, the \harqrc{} performs better in terms of convergence speed for $\bar{\tau}=4$ with $\lambda=1$, while with $\lambda=0.5$, the convergence is slightly faster for $\bar{\tau}=1$, and saturates for $\bar{\tau}\geq 3$. This is due to the fact that the trade-off between transmitting fresh information, rather than retransmitting old information, diminishes.

\begin{figure*}[!t]
\centering
\includegraphics[width=\linewidth]{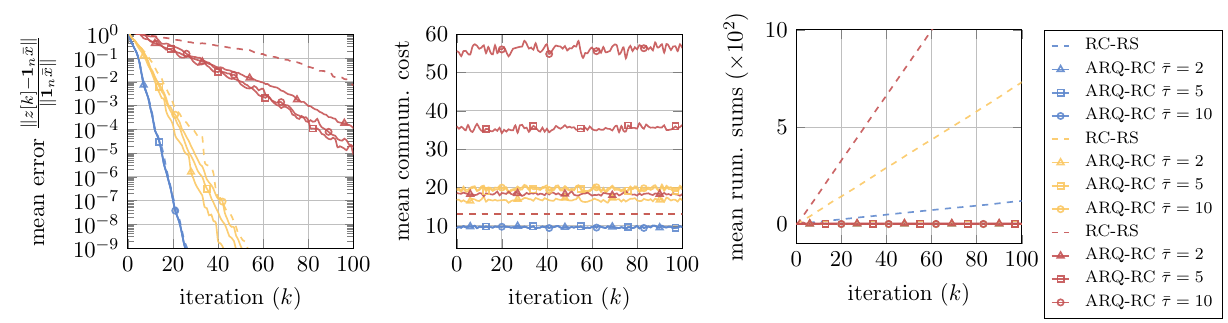}
\caption{Mean average consensus error at each iteration (left plot), mean communication cost captured by the mean number of information and feedback packets transmitted in the network at each iteration (middle plot), and mean running sums value at each iteration (right plot) of \arqrc{} ($\lambda=1$) compared to the RC-RS algorithm for different channel conditions. Good channel conditions ($q=0.1$) are shown in blue; medium channel conditions ($q=0.5$) are shown in yellow; poor channel conditions ($q=0.9$) are shown in red.}
\vspace{-5pt}
\label{fig:arq_comparisons}
\end{figure*}

\begin{figure*}[!t]
\centering
\includegraphics[width=\linewidth]{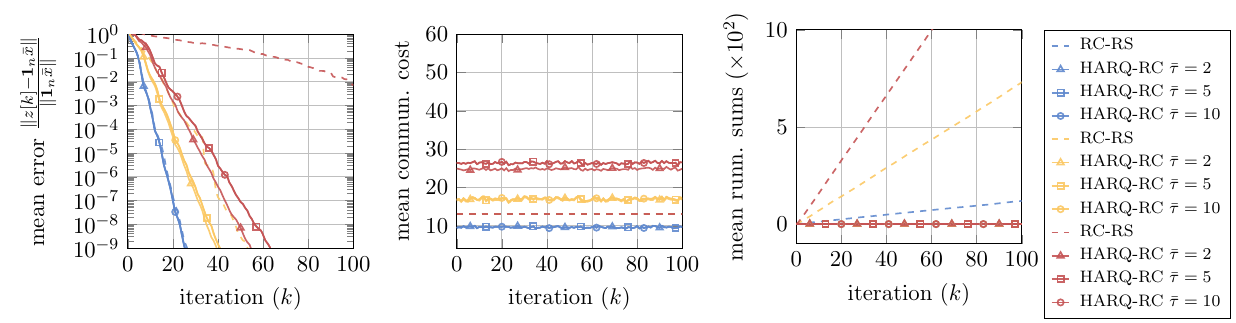}
\caption{Mean average consensus error at each iteration (left plot), mean communication cost captured by the mean number of information and feedback packets transmitted in the network at each iteration (middle plot), and mean running sums value at each iteration (right plot) of \harqrc{} ($\lambda=0.5$) compared to the RC-RS algorithm for different channel conditions. Good channel conditions ($q=0.1$) are shown in blue; medium channel conditions ($q=0.5$) are shown in yellow; poor channel conditions ($q=0.9$) are shown in red.}
\vspace{-5pt}
\label{fig:harq_comparisons}
\end{figure*}

\subsection{Convergence Speed vs. Communication Cost}
To get intuition regarding the convergence speed and communication cost of the proposed algorithm, we consider the same setup as before for the digraph shown in Fig.~\!\ref{fig:graph_numerical_example}. Specifically, we run $500$ simulations, in which we examine the mean average consensus error $\frac{\lVert {z}[k] - \mathbf{1}_{n} \overline{x}\rVert}{\lVert \mathbf{1}_{n} \overline{x}\rVert}$, the mean number of (re)transmissions sent over the network, and the mean values held in the local buffers for compensating the dropped information, at each time step $k$. Fig.~\!\ref{fig:arq_comparisons} and Fig.~\!\ref{fig:harq_comparisons} depict the aforementioned metrics for different channel conditions captured by the initial probability of error $q\in\{ 0.1, 0.5, 0.9\}$. In particular, Fig.~\!\ref{fig:arq_comparisons} compares the performance of the state-of-the-art algorithm \rcrs{} with the \arqrc{} algorithm with $\lambda=1$, while Fig.~\!\ref{fig:harq_comparisons} compares it with the \harqrc{} algorithm with $\lambda=0.5$. Under this setup, we additionally examine convergence speed and communication cost for different number of retransmission limits $\bar{\tau}$.

Clearly, for higher probability of initial packet error, \ie $q=\{ 0.5, 0.9 \}$, our proposed algorithm \harqrc{} (with any $\lambda=\{0.5,1\}$ and any retransmission limit $\bar{\tau}=\{2,5,10\}$) outperforms the \rcrs{} algorithm in terms of convergence speed, as shown on the left plots of Fig.~\!\ref{fig:arq_comparisons} and Fig.~\!\ref{fig:harq_comparisons}. Of course, this is in exchange of a higher number of transmitted packets as shown in the middle plots of Fig.~\!\ref{fig:arq_comparisons} and Fig.~\!\ref{fig:harq_comparisons}. However, one can notice the convergence speed-up when using the \harqrc{} (with $\lambda=0.5$), especially in poor channel conditions where $q=0.9$. This significant improvement in terms of convergence speed is due to the capability of the \harqrc{} algorithm to reduce the probability that a packet is received with errors in the subsequent transmissions as previously discussed in \eqref{eq:error_prob}. Interestingly, the mean number of sent packets within the network at each iteration $k$ using \harqrc{} algorithm, \ie around $18-20$ transmissions per iteration, is quite close to the one achieved with \rcrs{}, \ie around $14$ transmissions per iteration.

As aforementioned, implementing the \rcrs{} algorithm would require nodes to maintain extra variables to store run sum values that correspond to the information that is lost due to packet errors and drops. These values are accumulated at each iteration, and as discussed in \cite{hadjicostis2015robust}, they grow linearly in an unbounded manner with the number of iterations. This is shown in Fig.~\!\ref{fig:arq_comparisons} and Fig.~\!\ref{fig:harq_comparisons} for different packet error probabilities. Although this would not be an issue for nodes' memory capabilities (assuming that there exists a distributed termination mechanism), it would require nodes to transmit packets of increasing length, which would induce additional delays and packet drops due to the increased probability of packet errors. However, when using the \arqrc{} and \harqrc{} algorithm shown in Fig.~\!\ref{fig:arq_comparisons} and Fig.~\!\ref{fig:harq_comparisons}, respectively, nodes are not required to keep such variables as they have an acknowledgement feedback at each iteration that allows them to reset their local buffers immediately after releasing the dropped information back to their current local states. This allows nodes to keep their local buffers (the values of which are not transmitted, as opposed to the \rcrs{} algorithm) almost empty at each time step.

\section{Conclusions and Future Directions}\label{sec:conclusions}

\subsection{Conclusions}
In this paper, we proposed a distributed algorithm to achieve discrete-time asymptotic average consensus, in the presence of unreliable directed networks that induce packet errors. By incorporating \harq{} protocols into a ratio consensus-based algorithm, nodes are guaranteed to achieve average consensus despite the packet errors that may induce information delays due to retransmissions and packet drops. Moreover, by utilizing the acknowledgement feedback of the proposed algorithm, nodes are able to determine their out-degree, and know whether and when their transmitted packets arrived at their destination successfully. The convergence analysis of our proposed algorithm shows that nodes converge to the average consensus value with probability $1$ despite the packet errors. Finally, in the numerical evaluation we emphasize the convergence speed-up that can be achieved under different settings of the proposed \harqrc{} algorithm, and discuss the trade-off between convergence speed and communication cost, under different channel conditions, while also comparing against the current state-of-the-art algorithm.

\subsection{Future Directions}

As it is evident from the results of this work, there exists a trade-off between transmitting new information with a lower success probability and retransmitting previously failed information with a higher success probability. This is a well-known inherent tradeoff between freshness and reliability in networked systems \cite{ceran2019average}. Thus, a promising extension of this work to achieve faster convergence under specific communication and computation resource limitations, is to derive a packet transmission policy  based on the quality of the channel, so as to dynamically select the retransmission limit of the \harqrc{} algorithm.

Another interesting direction for future work is to investigate the effect of noisy communications on the exact average-consensus property established in this work. While distributed averaging over random network graphs in the presence of noise has been studied in \cite{LiWang2018}, incorporating communication noise into the proposed \hharq{}-based framework would allow one to characterize how noise, retransmissions, retransmission-induced delays, and residual packet losses jointly affect convergence to the exact initial average.

Another interesting extension is to study whether information-theoretic privacy guarantees, such as those considered for distributed averaging in \cite{Fereydounian2023TIT}, can be incorporated into the proposed \harqrc{} framework while retaining robustness to packet errors and retransmission-induced delays.

\begin{appendix}

\section{Proof of Theorem~\ref{theorem:arq_consensus}}
\label{appendix:theorem}


\begin{proof} Rather than relying on convergence of the rows of products of column-stochastic matrices, we directly analyze the ratio dynamics induced by the common augmented iteration for $\widetilde{\mathbf{x}}[k]$ and $\widetilde{\mathbf{y}}[k]$. Recall that the augmented dynamics can be written as
\begin{subequations}
\begin{align}
\widetilde{\mathbf{x}}[k+1] &= \Xi[k]\widetilde{\mathbf{x}}[k], \label{eq:app_augmented_dynamics_A} \\
\widetilde{\mathbf{y}}[k+1] &= \Xi[k]\widetilde{\mathbf{y}}[k], \label{eq:app_augmented_dynamics_B}
\end{align}    
\end{subequations}
where every $\Xi[k]\in\mathbb{R}_{\geq 0}^{\widetilde n\times\widetilde n}$ is column stochastic. The augmented initial conditions satisfy $\widetilde{x}_j[0]=x_j[0]$, $\widetilde{y}_j[0]=1$, for $v_j\in\mathcal V$, whereas
$\widetilde{x}_\ell[0]
=\widetilde{y}_\ell[0]=0$ for
$v_\ell\in
\widetilde{\mathcal V}\setminus\mathcal V$.
Since the same sequence of nonnegative matrices drives the $\widetilde{\mathbf{x}}$- and $\widetilde{\mathbf{y}}$-iterations,
and since both augmented vectors have zero initial values at all virtual nodes, the support of $\widetilde{\mathbf{x}}[k]$ is contained
in the support of $\widetilde{\mathbf{y}}[k]$ for every $k\in\mathbb{Z}_{\geq 0}$. In particular,
\begin{equation}
\widetilde{y}_{\ell}[k]=0 \ \Longrightarrow\  \widetilde{x}_{\ell}[k]=0,
\quad \forall\,v_\ell\in\widetilde{\mathcal V},\; k\in\mathbb{Z}_{\geq 0}.
\label{eq:app_support_property}
\end{equation}
This property does not imply that the states of the virtual nodes remain zero. Rather, virtual nodes may temporarily contain nonzero
$x$- and $y$-mass while representing information undergoing retransmission or return after a packet drop. Once such a virtual
node contains no $y$-mass, however, it necessarily contains no $x$-mass either. Moreover, because every actual node has a strictly positive self-loop weight $p_{jj}>0$, we have
\begin{equation}
y_j[k+1]\geq p_{jj}y_j[k]>0, \quad \forall\,v_j\in\mathcal V,\ k\in\mathbb{Z}_{\geq 0}.
\label{eq:app_positive_actual_y}
\end{equation}
Hence, the ratio $z_j[k]=x_j[k]/y_j[k]$ is well defined at every actual node and at every iteration. 
For the virtual nodes, define the active augmented-node set
\begin{equation}
\mathcal A[k] \triangleq \left\{
v_{\ell}\in\widetilde{\mathcal V}:
\widetilde{y}_{\ell}[k]>0
\right\}.
\label{eq:app_active_set}
\end{equation}
For every $v_{\ell}\in\mathcal A[k]$, let $\widetilde{z}_{\ell}[k] \triangleq {\widetilde{x}_{\ell}[k]}/{\widetilde{y}_{\ell}[k]}$. No ratio needs to be defined for an inactive virtual node, since \eqref{eq:app_support_property} implies that both its $x$- and
$y$-mass are zero.

\subsection{Ratio-Interval Invariance}

We first establish that the ratios of all active augmented nodes remain inside a nonexpanding interval. Define
\begin{equation}
m[k]
\triangleq
\min_{v_{\ell}\in\mathcal A[k]}
\widetilde{z}_{\ell}[k],
\qquad
M[k]
\triangleq
\max_{v_{\ell}\in\mathcal A[k]}
\widetilde{z}_{\ell}[k].
\label{eq:app_ratio_bounds}
\end{equation}
Consider any $v_{\ell}\in\mathcal A[k+1]$. From
\eqref{eq:app_augmented_dynamics_A}-\eqref{eq:app_augmented_dynamics_B}, we have
\begin{equation}
\widetilde{z}_{\ell}[k+1] = \sum_{v_i\in\mathcal A[k]}
\omega_{\ell i}[k]\widetilde{z}_i[k],
\label{eq:app_ratio_convex_combination}
\end{equation}
where $\omega_{\ell i}[k] \triangleq {\Xi_{\ell i}[k]\widetilde{y}_i[k]}/{\widetilde{y}_{\ell}[k+1]}$, $\omega_{\ell i}[k]\geq 0$, and $\sum_{v_i\in\mathcal A[k]}\omega_{\ell i}[k]=1$. Thus, each active ratio at time $k+1$ is a convex combination of active ratios at time $k$. Consequently,
\begin{equation}
m[k]\leq m[k+1]\leq M[k+1]\leq M[k],
\label{eq:app_interval_monotonicity}
\end{equation}
and therefore the ratio range $R[k]\triangleq M[k]-m[k]$ is nonincreasing.

\subsection{Contracting Communication Blocks}

We next show that communication blocks that contract the ratio range occur infinitely often almost surely. Let $\xi \triangleq \min_{v_j\in\mathcal V} {1}/{(1+\outdegree{j})} >0$ denote the smallest positive weight appearing in the original ratio-consensus matrix $P$. Since the underlying digraph $\mathcal G=(\mathcal V,\mathcal E)$ is strongly connected and each actual node has a self-loop, every pair of actual nodes is connected by a directed walk of exactly $n-1$ steps, after possibly padding a shorter directed path with self-loops. Hence,
\begin{equation}
P^{n-1}>0, \ [P^{n-1}]_{ji}\geq \xi^{n-1}, \quad \forall\,v_i,v_j\in\mathcal V.
\label{eq:app_positive_P_product}
\end{equation}
Consider consecutive nonoverlapping blocks of $L\triangleq\bar{\tau}+n-1$ iterations, with starting times $b_s\triangleq sL$, $s\in\mathbb Z_{\geq0}$. We call block $s$ a \emph{good block}, and denote the corresponding event by $\mathcal G_s$, if every packet transmission attempted during that block is successfully decoded. At any iteration, each directed communication link can carry at most one fresh packet and at most $\bar{\tau}$ retransmitted packets. Hence, there are at most $N_L\triangleq m(\bar{\tau}+1)L$ link-level packet-transmission opportunities during a block. Enumerate these potential transmission opportunities in chronological order by $\ell=1,\ldots,N_L$, and let $\mathcal S_{s,\ell}$ denote the event that either no packet is transmitted at the $\ell$th opportunity, or, if a packet is transmitted, it is successfully decoded. Then, $\mathcal G_s = \bigcap_{\ell=1}^{N_L}\mathcal S_{s,\ell}$. Under Assumption~\ref{ass:conditional_success}, whenever a packet is transmitted, its conditional probability of successful decoding, given the complete transmission history available immediately before that attempt, is at least $p_{\min}>0$. If no packet is transmitted at a particular opportunity, the corresponding event
$\mathcal S_{s,\ell}$ occurs with probability one. Therefore, the conditional probability of every event $\mathcal S_{s,\ell}$ is at least $p_{\min}$, uniformly over all admissible preceding transmission histories. Using the chain rule for conditional probabilities, we obtain
\begin{align}
\mathbb{P}\big( \mathcal G_s \mid \mathcal F_{b_s} \big)
&= \mathbb{P} \Big( \cap_{\ell=1}^{N_L}\mathcal S_{s,\ell} \mid \mathcal F_{b_s}
\Big)
\nonumber\\
&=
\prod_{\ell=1}^{N_L}3
\mathbb{P} \big(
\mathcal S_{s,\ell} \mid \mathcal F_{b_s},
\mathcal S_{s,1},\ldots, \mathcal S_{s,\ell-1}
\big)
\nonumber\\
&\geq
p_{\min}^{N_L}
\triangleq
\eta
>0,
\label{eq:good_block_probability}
\end{align}
where $\mathcal F_{b_s}$ denotes the information history available
at the beginning of block $s$. Importantly,
\eqref{eq:good_block_probability} does not require independence
between successive transmission attempts. The lower bound $\eta$ is uniform over all blocks. In particular, $\mathbb{P}\big(\mathcal G_s^c \mid \mathcal F_{b_s} \big) \leq 1-\eta$. Applying this bound recursively and using the tower property of conditional expectation yields, for every $N\in\mathbb Z_{>0}$,
\begin{equation}
\mathbb{P} \big( \cap_{\ell=s}^{s+N-1} \mathcal G_\ell^c \mid \mathcal F_{b_s} \big) \leq (1-\eta)^N.
\label{eq:no_good_block_probability}
\end{equation}
Since $(1-\eta)^N\rightarrow0$ as $N\rightarrow\infty$, the
probability that no good block occurs after any finite block index
$s$ is zero. Consequently,
\begin{equation}
\mathbb{P} \big( \mathcal G_s\ \text{occurs infinitely often} \big) =1.
\label{eq:good_blocks_infinitely_often}
\end{equation}
We now establish the contraction generated by each good block. Consider a good block starting at iteration $b_s$. Since no new delayed or dropped mass is generated during a good block, all mass already contained in the virtual delay and drop chains at time $b_s$ returns to actual nodes after at most $\bar{\tau}$ iterations. Thus, at $t_s\triangleq b_s+\bar{\tau}$, all virtual-node masses are zero and $\sum_{v_i\in\mathcal V}y_i[t_s]=n$. It follows that there exists at least one actual node, denoted by
$v_r$, such that
\begin{equation}
y_r[t_s]\geq 1.
\label{eq:app_large_y_node}
\end{equation}
During the remaining $n-1$ iterations of the good block, all freshly generated packets are successfully received without retransmission. Hence, the actual-node dynamics reduce to the ordinary ratio-consensus iterations
\begin{subequations}
\begin{align}
\mathbf{x}[t_s+n-1] &= P^{n-1}\mathbf{x}[t_s], \\
\mathbf{y}[t_s+n-1] &= P^{n-1}\mathbf{y}[t_s].
\label{eq:app_P_dynamics}
\end{align}  
\end{subequations}
For each actual node $v_j$, define
\begin{equation}
a_{ji}^{(s)} \triangleq \frac{ [P^{n-1}]_{ji}y_i[t_s] }{ y_j[t_s+n-1] }.
\label{eq:app_normalized_weights}
\end{equation}
The coefficients in \eqref{eq:app_normalized_weights} are nonnegative and satisfy $\sum_{v_i\in\mathcal V}a_{ji}^{(s)}=1$. Moreover, since the total $y$-mass is equal to $n$, we have
$y_j[t_s+n-1]\leq n$. Using \eqref{eq:app_positive_P_product} and \eqref{eq:app_large_y_node}, it follows that
\begin{equation}
a_{jr}^{(s)} \geq \frac{\xi^{n-1}}{n} \triangleq \beta >0, \quad \forall\,v_j\in\mathcal V.
\label{eq:app_common_weight}
\end{equation}
The ratio at each actual node at the end of the good block can therefore be written as
\begin{equation}
z_j[t_s+n-1] = \sum_{v_i\in\mathcal V} a_{ji}^{(s)}z_i[t_s].
\label{eq:app_ratio_good_block}
\end{equation}
Since $a_{jr}^{(s)}\geq\beta$, we can equivalently express
\eqref{eq:app_ratio_good_block} as
\begin{equation}
z_j[t_s+n-1] = \beta z_r[t_s] + (1-\beta)\hat{z}_j^{(s)},
\label{eq:app_common_component}
\end{equation}
where $\hat{z}_j^{(s)} \in \big[m[t_s],M[t_s]\big]$. It follows that the ratio range at the end of every good block
satisfies
\begin{equation}
R[b_s+L] \leq (1-\beta)R[t_s] \leq (1-\beta)R[b_s].
\label{eq:app_good_block_contraction}
\end{equation}
Notice that the virtual-node masses are again zero at the end of the good block, since every transmission during the block is successful.

\subsection{Almost-Sure Convergence to the Average}

Let $N(k)$ denote the number of completed good blocks up to iteration $k$. From \eqref{eq:app_interval_monotonicity} and
\eqref{eq:app_good_block_contraction}, $R[k] \leq (1-\beta)^{N(k)}R[0]$. By \eqref{eq:good_blocks_infinitely_often},
$N(k)\rightarrow\infty$ almost surely. Therefore,
\begin{equation}
\lim_{k\rightarrow\infty}R[k]=0 \quad\text{\emph{almost surely (a.s.)}}.
\label{eq:app_range_to_zero}
\end{equation}
Since $m[k]$ is nondecreasing, $M[k]$ is nonincreasing, and both are bounded by the minimum and maximum initial actual-node values, they
converge. Together with \eqref{eq:app_range_to_zero}, there exists a random variable $z^\star$ such that
\begin{equation}
\lim_{k\rightarrow\infty}m[k] = \lim_{k\rightarrow\infty}M[k] = z^\star \quad\text{\emph{a.s.}}.
\label{eq:app_common_ratio_limit}
\end{equation}
Every actual-node ratio satisfies $m[k]\leq z_j[k]\leq M[k]$,~$\forall\,v_j\in\mathcal V$, and hence
\begin{equation}
\lim_{k\rightarrow\infty}z_j[k] = z^\star, \quad \forall\,v_j\in\mathcal V, \quad\text{\emph{a.s.}}.
\label{eq:app_actual_common_limit} 
\end{equation}
Now, it remains to identify $z^\star$. Since each $\Xi[k]$ is column
stochastic, the total $x$- and $y$-masses of the augmented system are
preserved, \ie
\begin{subequations}
\begin{align}
\mathbf{1}_{\widetilde n}^{\top}
\widetilde{\mathbf{x}}[k]
&=
\mathbf{1}_{\widetilde n}^{\top} \widetilde{\mathbf{x}}[0] = \sum_{v_i\in\mathcal V}x_i[0],
\label{eq:app_x_mass_conservation}\\
\mathbf{1}_{\widetilde n}^{\top} \widetilde{\mathbf{y}}[k]
&= \mathbf{1}_{\widetilde n}^{\top} \widetilde{\mathbf{y}}[0] = n.
\label{eq:app_y_mass_conservation}
\end{align}   
\end{subequations}
Using \eqref{eq:app_support_property}, only active augmented nodes
contribute to these sums. Therefore,
\begin{align}
\left| \sum_{v_i\in\mathcal V}x_i[0] - nz^\star \right|
&= \left| \sum_{v_{\ell}\in\mathcal A[k]} \widetilde{y}_{\ell}[k] \left( \widetilde{z}_{\ell}[k]-z^\star \right) \right| \nonumber\\
&\leq n \max_{v_{\ell}\in\mathcal A[k]} \left| \widetilde{z}_{\ell}[k]-z^\star \right|.
\label{eq:app_identify_limit}
\end{align}
From \eqref{eq:app_common_ratio_limit}, the right-hand side of
\eqref{eq:app_identify_limit} converges to zero almost surely.
Consequently,
\begin{equation}
z^\star = \frac{1}{n} \sum_{v_i\in\mathcal V}x_i[0] = \overline{x}.
\label{eq:app_limit_average}
\end{equation}
Combining \eqref{eq:app_actual_common_limit} and \eqref{eq:app_limit_average}, we finally obtain
\begin{equation}
\mathbb{P} \bigg( \lim_{k\rightarrow\infty}z_j[k] = \frac{1}{n}
\sum_{v_i\in\mathcal V}x_i[0], \ \forall\,v_j\in\mathcal V \bigg) =1,
\end{equation}
which proves Theorem~1.
\end{proof}

\section{Convergence Rate}
\label{app:convergence_rate}

In this appendix, we derive a convergence-rate bound for the \harqrc{} algorithm. In contrast to analyzing the range of the raw augmented state, which evolves according to a column-stochastic iteration and does not, in general, reach consensus, we directly
analyze the ratio range introduced in Appendix~\ref{appendix:theorem}.

Recall the ratio range defined by
\begin{equation}
R[k] \triangleq M[k]-m[k],
\label{eq:appC_ratio_range}
\end{equation}
where $m[k]$ and $M[k]$ denote, respectively, the minimum and maximum ratios over the active augmented nodes at iteration $k$.
From~\eqref{eq:app_interval_monotonicity}, $R[k]$ is nonincreasing. Furthermore, recall the block length $L \triangleq \bar{\tau}+n-1$,
and the good-block event $\mathcal G_s$ introduced in Appendix~\ref{appendix:theorem}. For the rate analysis, consider consecutive nonoverlapping blocks with starting times $b_s\triangleq sL$, $s\in\mathbb{Z}_{\geq 0}$. By~\eqref{eq:good_block_probability}, there exists a constant $\eta>0$ such that
\begin{align*}
\mathbb{P} \big( \mathcal G_s \mid \mathcal F_{b_s} \big) \geq \eta, \quad \forall\,s\in\mathbb{Z}_{\geq 0},
\end{align*}
where $\mathcal F_{b_s}$ denotes the information available at the beginning of block $s$. A conservative choice established in Appendix~\ref{appendix:theorem} is
\begin{equation}
\eta = p_{\min}^{N_L},
\quad N_L = m(\bar{\tau}+1)L.
\label{eq:appC_eta}
\end{equation}
Recall also from~\eqref{eq:app_common_weight} that
\begin{equation}
\beta \triangleq \frac{\xi^{n-1}}{n} >0, \quad \xi \triangleq \min_{v_j\in\mathcal V} \frac{1}{1+d_j^+}.
\label{eq:appC_beta}
\end{equation}
According to~\eqref{eq:app_good_block_contraction}, whenever
$\mathcal G_s$ occurs, $R[b_{s+1}] \leq (1-\beta)R[b_s]$.
If $\mathcal G_s$ does not occur, the monotonicity of the ratio
interval gives $R[b_{s+1}] \leq R[b_s]$.

\begin{prop}[Expected convergence rate]
\label{prop:expected_rate}
Consider the assumptions of Theorem~1 and let the ratio range
$R[k]$ be defined as in~\eqref{eq:appC_ratio_range}. Then, for all
$s\in\mathbb Z_{\geq0}$,
\begin{equation}
\mathbb{E}\big\{R[b_s]\big\} \leq \rho^sR[0],
\label{eq:expected_block_rate}
\end{equation}
where $\rho\triangleq1-\eta\beta\in(0,1)$. Consequently, for every $k\in\mathbb Z_{\geq0}$,
\begin{equation}
\mathbb{E}\big\{ R[k] \big\} \leq \rho^{\lfloor k/L\rfloor}R[0].
\label{eq:expected_iteration_rate}
\end{equation}
\end{prop}

\begin{proof}
Let $\mathcal{I}_s \triangleq \mathbf{1}_{\{\mathcal G_s\}}$ denote the indicator of the good-block event. Combining the
contraction during a good block with the nonexpansion of the ratio range during any other block gives the pathwise inequality
\begin{equation}
R[b_{s+1}] \leq \left(1-\beta \mathcal{I}_s\right)R[b_s].
\label{eq:random_block_contraction}
\end{equation}
Iterating~\eqref{eq:random_block_contraction} over the first $s$ blocks yields
\begin{align}
R[b_s] \leq R[0] \prod_{\ell=0}^{s-1} \left(1-\beta \mathcal{I}_\ell\right) = R[0](1-\beta)^{S_s},
\label{eq:pathwise_good_block_bound}
\end{align}
where
\begin{align*}
S_s \triangleq \sum_{\ell=0}^{s-1}\mathcal{I}_\ell = \sum_{\ell=0}^{s-1}
\mathbf{1}_{\{\mathcal G_\ell\}}
\end{align*}
is the number of good blocks completed among the first $s$
blocks.
Define $Y_s \triangleq (1-\beta)^{S_s}$. Since blocks $0,\ldots,s-1$ have been completed by time $b_s$, the random variable $S_s$, and hence $Y_s$, is measurable with
respect to $\mathcal F_{b_s}$. Moreover, taking conditional expectation of $Y_{s+1} = Y_s(1-\beta)^{\mathcal{I}_s}$ with respect to $\mathcal F_{b_s}$ gives
\begin{align}
\mathbb{E}\big\{ Y_{s+1} \mid \mathcal F_{b_s} \big\}
&= Y_s
\mathbb{E} \big\{ (1-\beta)^{\mathcal{I}_s} \mid \mathcal F_{b_s} \big\}
\nonumber\\
&= Y_s \big[ 1 - \beta \mathbb{P}\big( \mathcal G_s \mid \mathcal F_{b_s} \big)
\big]
\nonumber\\
&\leq
(1-\eta\beta)Y_s = \rho Y_s,
\label{eq:Ys_conditional_contraction}
\end{align}
where~\eqref{eq:good_block_probability} was used in the
inequality. Taking expectations in~\eqref{eq:Ys_conditional_contraction}
and using the tower property yields
\begin{align*}
\mathbb{E}\big\{Y_{s+1}\big\}
\leq \rho\,\mathbb{E}\big\{Y_s\big\}.
\end{align*}
Since $S_0=0$ and hence $Y_0=1$, recursive application gives
\begin{equation}
\mathbb{E}\big\{Y_s\big\}
\leq \rho^s.
\label{eq:Ys_expectation}
\end{equation}
Combining~\eqref{eq:pathwise_good_block_bound} and
\eqref{eq:Ys_expectation}, we obtain
\begin{align*}
\mathbb{E}\big\{R[b_s]\big\} \leq R[0]\mathbb{E}\big\{ Y_s \big\}
\leq \rho^sR[0],
\end{align*}
which proves the first claim. For an arbitrary iteration $k$, let $s=\left\lfloor\frac{k}{L}\right\rfloor$. Since $b_s=sL\leq k$ and the ratio range is nonincreasing,
\begin{align*}
R[k]\leq R[b_s].
\end{align*}
Taking expectations and applying the preceding block-level bound
gives
\begin{align*}
\mathbb{E}\{R[k]\}
\leq
\rho^{\lfloor k/L\rfloor}R[0],
\end{align*}
which establishes the second claim.
\end{proof}

The initial ratio range has the particularly simple form
\begin{equation}
R[0] = \max_{v_i\in\mathcal V}x_i[0] - \min_{v_i\in\mathcal V}x_i[0],
\label{eq:appC_initial_range}
\end{equation}
since initially $y_i[0]=1$ for all actual nodes and all virtual
nodes contain zero mass. We next relate the ratio range directly to the error with respect
to the exact average. From the mass-preservation properties in
\eqref{eq:app_x_mass_conservation} and
\eqref{eq:app_y_mass_conservation}, and from the support property
in~\eqref{eq:app_support_property}, we have at every iteration
\begin{equation}
\overline{x} = \frac{1}{n} \sum_{v_\ell\in\mathcal A[k]} \widetilde{y}_{\ell}[k] \widetilde{z}_{\ell}[k].
\label{eq:appC_average_convex_combination}
\end{equation}
Since ${\widetilde{y}_{\ell}[k]}/{n}\geq0$ and $\sum_{v_\ell\in\mathcal A[k]} {\widetilde{y}_{\ell}[k]}/{n}=1$, the average $\overline{x}$ is a convex combination of the active ratios. Hence, $m[k] \leq \overline{x} \leq M[k]$. Therefore, for every actual node $v_j\in\mathcal V$,
\begin{equation}
\left| z_j[k]-\overline{x} \right| \leq R[k].
\label{eq:appC_node_error_range}
\end{equation}

\begin{cor}[Expected average-consensus error]
\label{cor:expected_consensus_error}
Under the assumptions of Proposition~\ref{prop:expected_rate},
\begin{equation}
\mathbb{E}\big\{ \left\| \mathbf z[k]-\overline{x}\mathbf 1_n
\right\|_{\infty} \big\}
\leq
\rho^{\left\lfloor k/L\right\rfloor}R[0].
\label{eq:appC_inf_error}
\end{equation}
Moreover,
\begin{equation}
\mathbb{E} \big\{ \left\|
\mathbf z[k]-\overline{x}\mathbf 1_n
\right\|_2
\big\}
\leq
\sqrt{n}\,
\rho^{\left\lfloor k/L\right\rfloor}R[0].
\label{eq:appC_two_error}
\end{equation}
\end{cor}

\begin{proof}
From~\eqref{eq:appC_node_error_range}, $\left\| \mathbf z[k]-\overline{x}\mathbf 1_n \right\|_{\infty} \leq R[k]$. Taking expectations and applying \eqref{eq:expected_iteration_rate} gives \eqref{eq:appC_inf_error}. Equation~\eqref{eq:appC_two_error} then follows from $\|\mathbf v\|_2 \leq \sqrt{n}\|\mathbf v\|_{\infty}$.
\end{proof}
Using~\eqref{eq:appC_eta} and~\eqref{eq:appC_beta}, a fully
explicit, although conservative, block-contraction factor is
\begin{equation}
\rho = 1- \frac{\xi^{n-1}}{n} p_{\min}^{m(\bar{\tau}+1)(\bar{\tau}+n-1)}.
\label{eq:appC_explicit_rho}
\end{equation}
Thus, the expected consensus error converges geometrically with
respect to the number of communication blocks. It is also useful to associate an effective per-iteration contraction factor with~\eqref{eq:expected_iteration_rate}.
Define $\rho_{\mathrm{it}}\triangleq\rho^{1/L}=(1-\eta\beta)^{1/L}$. Then, the convergence is geometric on the block time scale, with
effective per-iteration factor $\rho_{\mathrm{it}}<1$.

\begin{cor}[High-probability error bound]
\label{cor:high_probability}
For every $\epsilon>0$ and every
$k\in\mathbb{Z}_{\geq0}$,
\begin{equation}
\mathbb{P} \big( \left\| \mathbf z[k]-\overline{x}\mathbf 1_n \right\|_{\infty} >\epsilon \big) \leq \min\left\{ 1,\, \frac{R[0]}{\epsilon} \rho^{\left\lfloor k/L\right\rfloor} \right\}.
\label{eq:appC_high_probability}
\end{equation}
\end{cor}
\begin{proof}
By~\eqref{eq:appC_node_error_range},
\begin{equation}
\mathbb{P} \big( \left\| \mathbf z[k]-\overline{x}\mathbf 1_n \right\|_{\infty} >\epsilon \big) \leq \mathbb{P} \big( R[k]>\epsilon \big).
\end{equation}
The result follows from Markov's inequality and
\eqref{eq:expected_iteration_rate}.
\end{proof}

Finally, we provide a bound on the expected time required to reach
a prescribed consensus accuracy.
\begin{prop}[Expected $\epsilon$-consensus time]
\label{prop:epsilon_time}
For $\epsilon\in(0,R[0])$, define the $\epsilon$-consensus
time as
\begin{equation}
T_\epsilon
\triangleq
\inf
\left\{
k\in\mathbb{Z}_{\geq0}:
\left\|
\mathbf z[k]-\overline{x}\mathbf 1_n
\right\|_{\infty}
\leq\epsilon
\right\}.
\label{eq:appC_epsilon_time}
\end{equation}
Let
\begin{equation}
r_\epsilon
\triangleq
\left\lceil
\frac{
\log\!\left(R[0]/\epsilon\right)
}{
-\log(1-\beta)
}
\right\rceil.
\label{eq:appC_required_good_blocks}
\end{equation}
Then,
\begin{equation}
\mathbb{E}\!\left\{T_\epsilon\right\}
\leq
\frac{L}{\eta}\,r_\epsilon.
\label{eq:appC_expected_epsilon_time}
\end{equation}
\end{prop}
\begin{proof}
Let
\begin{equation}
S_s
\triangleq
\sum_{\ell=0}^{s-1}
\mathbf{1}_{\{\mathcal G_\ell\}}
\label{eq:appC_number_good_blocks}
\end{equation}
denote the number of good blocks completed during the first $s$
blocks. Since the ratio range cannot increase during a bad block
and contracts by at least a factor of $(1-\beta)$ during every good
block, $R[b_s] \leq (1-\beta)^{S_s}R[0]$. Therefore, once $S_s\geq r_\epsilon$, $R[b_s]\leq(1-\beta)^{r_\epsilon}R[0]\leq\epsilon$. From~\eqref{eq:appC_node_error_range}, all actual-node ratios
are then within $\epsilon$ of the exact average.
Let $H_r$ denote the number of blocks required to observe $r$
good blocks. Since at the beginning of every block the conditional
probability of a good block is at least $\eta$, the waiting time
between two successive good blocks has conditional tail bounded by
\begin{equation}
\mathbb{P}(W>u) \leq (1-\eta)^u, \quad u\in\mathbb{Z}_{\geq0}.
\end{equation}
Consequently, $\mathbb{E}\{W\}\leq\sum_{u=0}^{\infty}(1-\eta)^u={1}/{\eta}$. Hence, $\mathbb{E}\{H_r\} \leq {r}/{\eta}$. Since each block contains $L$ iterations and
$T_\epsilon\leq LH_{r_\epsilon}$, the result in
\eqref{eq:appC_expected_epsilon_time} follows.
\end{proof}
\begin{rem}
\label{rem:conservative_rate}
The bounds in
\eqref{eq:appC_explicit_rho} and
\eqref{eq:appC_expected_epsilon_time} are deliberately conservative.
They are based on the sufficient good-block event introduced in
Appendix~\ref{appendix:theorem}, which requires all packet transmissions involved in the
block to be successful. The actual convergence can therefore be
considerably faster, since contraction may also occur during blocks
that do not satisfy this sufficient event. The purpose of the bounds
is to establish an explicit geometric convergence rate and a finite
expected $\epsilon$-consensus time under the same assumptions used
for the almost-sure convergence proof.
\end{rem}
\begin{rem}[Dependence on the \harq{} reliability model]
\label{rem:harq_rate}
Under the \harq{} model, the conditional probability of successful decoding at retransmission index $r$ is $1-q_{ij}\lambda^r$. Since $1-q_{ij}\lambda^r \geq 1-q_{ij} \geq 1-q_{\max}$, the general bound above applies with $p_{\min}=1-q_{\max}$. This choice ignores the improvement in decoding probability provided by successive \harq{} retransmissions and is therefore conservative. A tighter rate can be obtained by replacing the uniform good-block probability $\eta$ by a lower bound that explicitly accounts for the retransmission-dependent success probabilities $1-q_{ij}\lambda^r$. The convergence arguments remain unchanged. Only the value of $\eta$, and consequently the contraction factor $\rho=1-\eta\beta$, is modified.
\end{rem}

\end{appendix}

\bibliographystyle{IEEEtran} 
\bibliography{references}

\end{document}